\documentclass[lettersize]{vldb}

\usepackage{graphicx}
\usepackage{amsmath}
\usepackage{amssymb}
\usepackage{url}
\usepackage{xspace}
\usepackage{balance}
\usepackage{algorithm}
\usepackage[noend]{algpseudocode}
\usepackage{multirow}
\usepackage{subfig}
\usepackage{setspace}
\usepackage{cite}
\usepackage{hyperref}
\usepackage{enumitem}
\usepackage{fixltx2e}
\usepackage{xcolor}
\usepackage{times}

\graphicspath{{./graphs/}}

\newtheorem{thm}{\textbf{Theorem}}
\newtheorem{defn}{\textbf{Definition}}
\newtheorem{lem}{\textbf{Lemma}}

\newtheorem{cor}{\textbf{Corollary}}

\newtheorem{example}{\textbf{Example}}
\newtheorem{prob}{\textbf{Problem}}

\title{NetClus: A Scalable Framework for Locating Top-K Sites for Placement of Trajectory-Aware Services}

\numberofauthors{1}
\author{
\alignauthor
Shubhadip Mitra$^1$ ~
Priya Saraf$^1$ ~
Richa Sharma$^1$ ~
Arnab Bhattacharya$^1$ ~
Sayan Ranu$^2$
	$^1$\affaddr{Dept. of Computer Science and Engineering, Indian Institute of Technology, Kanpur, India}\\
	$^2$\affaddr{Dept. of Computer Science and Engineering, Indian Institute of Technology, Delhi, India}\\
	\email{$^1$\{smitr,priyas,richash,arnabb\}@cse.iitk.ac.in ~ $^2$sayanranu@cse.iitd.ac.in}
}

\newcommand{\ipopt}{\textsc{OPT}\xspace}

\newcommand{\incg}{\textsc{Inc-Greedy}\xspace}

\newcommand{\nc}{\textsc{NetClus}\xspace}
\newcommand{\topscost}{\textsc{Tops-Cost}\xspace}
\newcommand{\topscap}{\textsc{Tops-Capacity}\xspace}
\newcommand{\fmnc}{\textsc{FM-NetClus}\xspace}

\newcommand{\eps}{\ensuremath{\gamma}\xspace}
\newcommand{\epsinc}{\ensuremath{(1 + \gamma)}\xspace}
\newcommand{\ds}{\ensuremath{d}\xspace}
\newcommand{\dr}{\ensuremath{d_r}\xspace}
\newcommand{\pref}{\ensuremath{\psi}\xspace}
\newcommand{\tops}{TOPS\xspace}
\newcommand{\topsone}{\textsc{Tops1}\xspace}
\newcommand{\topstwo}{\textsc{Tops2}\xspace}
\newcommand{\topsthree}{\textsc{Tops3}\xspace}
\newcommand{\topsfour}{\textsc{Tops4}\xspace}

\newcommand{\topsc}{TOPS-Cluster\xspace}
\newcommand{\traj}{\ensuremath{\mathcal{T}}\xspace}

\newcommand{\s}{\ensuremath{\mathcal{S}}\xspace}

\newcommand{\q}{\ensuremath{\mathcal{Q}}\xspace}
\newcommand{\jd}{\ensuremath{\alpha}\xspace}
\newcommand{\tc}{\ensuremath{TC}\xspace}
\newcommand{\cc}{\ensuremath{CC}\xspace}
\newcommand{\stc}{\ensuremath{SC}\xspace}
\newcommand{\cl}{\ensuremath{\mathcal{CL}}\xspace}
\newcommand{\tl}{\ensuremath{\mathcal{TL}}\xspace}
\newcommand{\inst}{\ensuremath{\mathcal{I}}\xspace}

\newcommand{\opt}{\textsc{Opt}\xspace}

\newcommand{\inc}{\textsc{IncG}\xspace}
\newcommand{\incfm}{\textsc{FMG}\xspace}
\newcommand{\ncfm}{\textsc{FMNetClus}\xspace}

\newcommand{\bs}{Beijing-Small\xspace}
\newcommand{\bl}{Beijing\xspace}

\newcommand{\subfigwidth}{0.48\columnwidth}

\newcommand{\figcaption}[1]{\vspace*{-3mm}\caption{#1}\vspace*{-4mm}}
\newcommand{\tabcaption}[1]{\vspace*{-3mm}\caption{#1}\vspace*{-4mm}}

\newcommand{\moveups}{\vspace*{-1mm}}
\newcommand{\moveup}{\vspace*{-2mm}}
\newcommand{\moveupb}{\vspace*{-3mm}}

\begin{document}

\maketitle

\begin{abstract} 
	Facility location queries identify the best locations to set up new
	facilities for providing service to its users. Majority of the existing
	works in this space assume that the user locations are \emph{static}. 
	Such limitations are too restrictive for planning many modern real-life
	services such as fuel stations, ATMs, convenience stores, cellphone
	base-stations, etc. that are widely accessed by mobile users.  The
	placement of such services should, therefore, factor in the \emph{mobility
	patterns} or \emph{trajectories} of the users rather than simply their
	static locations.  In this work, we introduce the \emph{\tops
	(Trajectory-Aware Optimal Placement of Services)} query that locates the
	best $k$ sites on a road network.  The aim is to optimize a wide class of
	objective functions defined over the user trajectories. We show that the
	problem is NP-hard and even the greedy heuristic with an approximation
	bound of $1-1/e$ fails to scale on urban-scale datasets.  To overcome this
	challenge, we develop a \emph{multi-resolution clustering} based indexing
	framework called \emph{\nc}. Empirical studies on real road network
	trajectory datasets show that \nc offers solutions that are comparable in
	terms of quality with those of the greedy heuristic, while having practical
	response times and low memory footprints.  Additionally, the \nc framework
	can absorb dynamic updates in mobility patterns, handle constraints such as
	site-costs and capacity, and existing services, thereby providing an
	effective solution for modern urban-scale scenarios.
\end{abstract}

\keywords{Spatio-temporal databases; Facility location queries; Optimal location queries; Road networks; Trajectory-aware services;}

\section{Introduction and Motivation}
\label{sec:Intro}

\emph{Facility Location queries}  
(or \emph{Optimal location (OL) queries}) in a road network aim to identify the best
locations to set up new facilities with respect to a given service
\cite{tops_icde,du2005optimal, ghaemi2010optimal, xiao2011optimal, chen2014efficient,li2016mining}.
Examples include setting up new retail stores, gas stations, or cellphone base
stations. OL queries also find applications in various spatial decision support
systems, resource planning and infrastructure management\cite{myinfocom,
PeopleInMotionCDR}.

\begin{figure}[t]
  \centering
  \includegraphics[width=0.60\columnwidth]{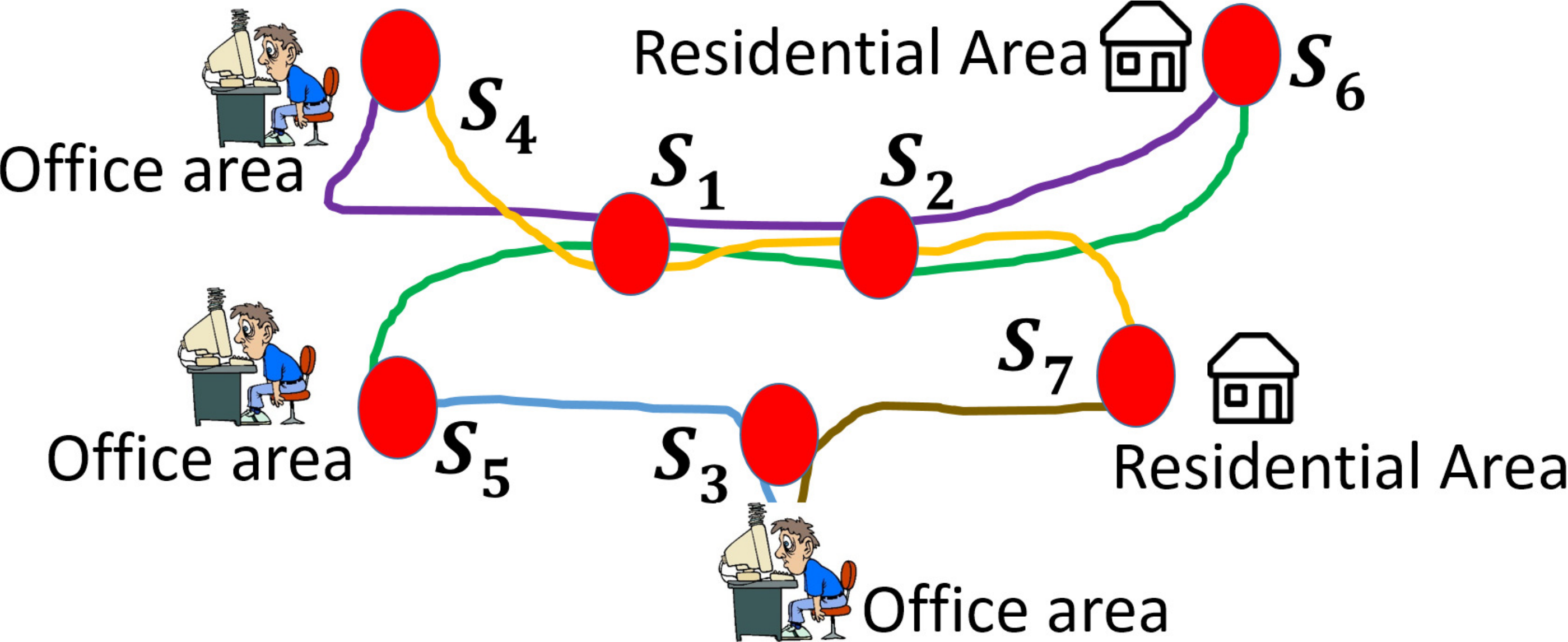}
  \figcaption{Illustration of the need of trajectory-aware querying for optimal
  locations. Each line represents a separate user trajectory.}
  \label{fig:motivation}
  \moveup
\end{figure}

With growing applications of data-driven location-based systems, the importance of OL queries is well-recognized in the database community  \cite{zhang2006progressive,xiao2011optimal,chen2014efficient}.
However, most of the existing works assume the users to be \emph{fixed} or \emph{static}.  Such an
assumption is often too prohibitive. For example, services such as gas
stations, ATMs, bill boards, traffic monitoring systems, etc. are widely
accessed by users while commuting.  Further, it is common for many users to make their daily purchases while returning from their offices. Consequently, the placement of facilities
for these services require taking into consideration the mobility patterns (or
trajectories) of the users rather than their static locations. We refer to such
services as \emph{trajectory-aware services}.  Formally, a \emph{trajectory} is
a sequence of location-time coordinates that lie on the path of a moving user. Note that trajectories strictly generalize the static
users' scenario because static users can always be modeled as trajectories with
a single user location. Thus, trajectories capture user patterns more
effectively.
Such trajectory data are commonly available from GPS traces \cite{RideSharing}
or CDR (Call Detail Records) data \cite{PeopleInMotionCDR} recorded through
cellphones, social network check-ins, etc. 
Recently, there have been many works
in the area of trajectory data analytics \cite{edwp,DTW,LCSS, DISSIM, ERP,
lee2007trajectory, kalnis2005discovering, jeung2008discovery, li2010swarm}.

To illustrate the need for trajectory-aware optimal location queries, consider
Fig.~\ref{fig:motivation}. There are 5 locations that are either homes or
offices and 5 trajectories of users commuting across them. A company wants to
open 2 new gas stations.  For simplicity, we assume that a trajectory of a user
is \emph{satisfied} if it passes through at least one gas station. If
only the static locations are considered, i.e., any two out of the five office
and residential areas are to be selected, no combination of gas stations would
satisfy all the users. In contrast, if we factor in the mobility of the users,
choosing $S_1$ and $S_3$ as the installation locations satisfies all the
trajectories of the users.

Note that it is not enough to simply look at trajectory counts in each possible
installation location and then choose the two most frequent ones ($S_1$ and
$S_2$). The combination may not be effective since a large number of
trajectories between them can be common, thereby reducing each other's
utilities.

In this work, we formalize the problem of OL queries over trajectories of users
in road networks. We refer to this as  \emph{\tops (Trajectory-aware Optimal
Placement of Services)} query. Given a set of user trajectories $\mathcal{T}$,
and a set of candidate sites $\mathcal{S}$ over a road network that can host
the services, \tops query with input parameters $k$ and preference function $\pref$ seeks to report the
best $k$ sites $\mathcal{Q} \subseteq \mathcal{S}$ maximizing  a \emph{utility function} that is defined over the preference function $\pref$ which captures  how a particular candidate site is \emph{preferred} by a given  trajectory.

Facility location queries with respect to trajectories have been studied by a number of previous
works \cite{berman1992optimal, berman1995locatingDiscretionary,
berman1995locatingMain, berman2002generalized,berman1998flow, li2013trajectory}.
Although the formulations are not identical, the common eventual goal is to
identify the best sites.  However, all these works remain limited to a
theoretical exercise and cannot be applied in a real-life scenario due to a
number of issues as explained next.

$\bullet$ \textbf{Data-based mobility model:} Existing techniques are neither
based on real trajectories nor on real road networks \cite{berman1992optimal,
berman1995locatingDiscretionary, berman1995locatingMain,
berman2002generalized,berman1998flow}. They base their
solutions on simplistic assumptions such as traveling in shortest paths and
synthetic road networks. It is well known that the shortest path assumption does
not hold in real life \cite{utraj}. In many works, the distances are \emph{not} computed over the road network but approximated using some spatial distance measure such as $L_2$ norm. Our framework is the first to study \tops on
real trajectories over real road networks.

$\bullet$ \textbf{Generic framework:} We develop the first generic framework to
answer \tops queries  across a wide family of preference functions $\pref$ that are non-increasing w.r.t. the distance between any pair of trajectory and candidate site. The proposed framework
encompasses many of the existing formulations and also considers other practical factors such as
capacity constraints, site-costs, dynamic updates, etc. 

$\bullet$ \textbf{Scalability:} The state-of-the-art technique for a basic version of \tops query
\cite{berman1995locatingMain} requires prohibitively large memory.
Consequently, it fails to scale on urban-scale datasets (Further details in Sec.~\ref{sec:exp}). Hence, a scalable
framework for \tops query is a basic necessity. In addition, all OL queries
including \tops are typically
used in an interactive fashion by varying the various parameters such as $k$ and
the coverage threshold $\tau$ \cite{chen2014efficient}. Moreover, in certain ventures, such as deployment of mobile ATM vans(\url{https://goo.gl/WjSPvx}), real-time answers are need based on current trajectory patterns. Hence, practical
response time with the ability to absorb data updates is critical. This factor has been completely ignored in the
existing
works. 

$\bullet$ \textbf{Extensive benchmarking:} Since \tops queries and their
variants are NP-hard, heuristics have been proposed. How do their effectiveness
vary across road-network topologies? Are these heuristics biased towards certain
specific parameter settings?  The existing techniques are generally silent on
these questions.  We, on the other hand, perform benchmarking that is grounded
to reality by extensively studying the performance of \tops across multiple
major city topologies and other important parameters.

To summarize, the proposed framework is the first practical and generic solution
to address \tops queries. Fig.~\ref{fig:flowchart} depicts the top-level flow
diagram of our solution.  Given the raw GPS traces of user movements, they are
map-matched \cite{map1} to the corresponding road network.  Using the
map-matched trajectories, a multi-resolution clustering of the road network is
built to construct the index structure \emph{\nc}. Indexed views of both the
candidate sites and the trajectories are maintained in a compressed format at
various granularities. This completes the offline phase.  In the online phase,
given the query parameters, the optimum clustering resolution to answer the
query is identified, and the corresponding views of the trajectories and road
network are analyzed to retrieve the best $k$ sites for facility locations.

\begin{figure}[t]
  \centering
  \includegraphics[width=0.80\columnwidth]{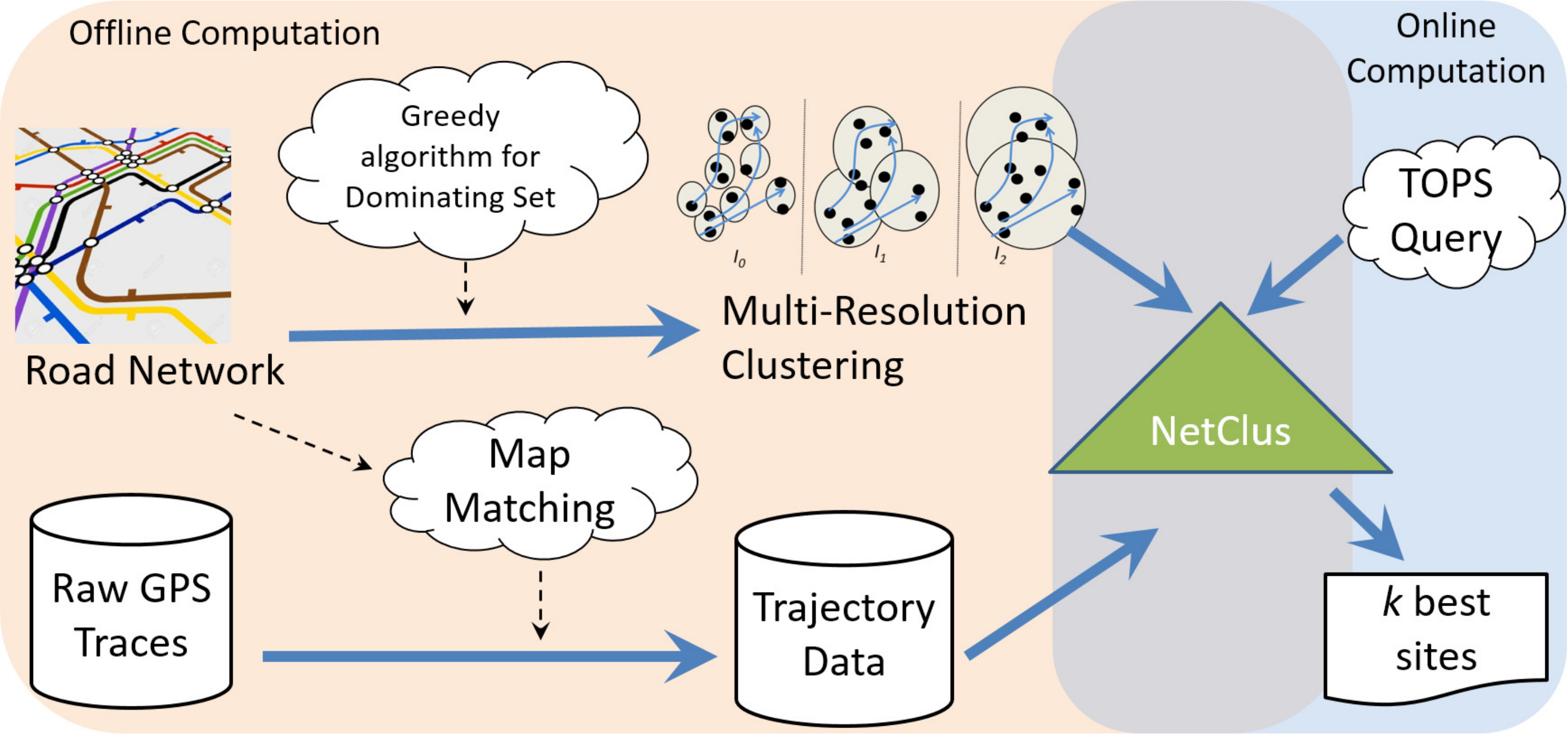}
  \figcaption{Flowchart of \tops querying framework.}
  \label{fig:flowchart}
  \moveup
\end{figure}

The major contributions of our work are as follows:

\begin{enumerate}[nosep,leftmargin=*]

	\item We propose a highly generic trajectory-aware optimal placement of
		services problem, \tops, that can handle the wide class of
 preference functions $\pref$ that are non-increasing w.r.t. the distance between any pair of trajectory and candidate site	(Sec.~\ref{sec:formulation}).
		

	\item  \tops is NP-hard, and even the greedy approach does not scale in city-scale datasets	(Sec.~\ref{sec:algo}). To overcome this bottleneck, we design a
		\emph{multi-resolution clustering-based} index structure called
		\emph{\nc} that generates a \emph{compressed} representation of the road
		network and the set of trajectories (Sec.~\ref{sec:offline} and
		Sec.~\ref{sec:online}).
The solutions returned by \nc framework have bounded quality guarantees.


	      \item The proposed \tops formulation is capable of absorbing dynamic updates (Sec.~\ref{sec:updates}) and manage realistic constraints such as cost, capacity, and existing services. 
	  (Sec.~\ref{sec:variants}).

	\item Extensive experiments on real datasets show that \nc offers solutions that are comparable in terms of quality with those of the greedy heuristic, while having practical response times and fairly low memory footprints 	
		 (Sec.~\ref{sec:exp}).
\end{enumerate}

\section{The \tops Problem}
\label{sec:formulation}

\begin{table}[t]
\scriptsize
\centering
\begin{tabular}{cl}
\hline
\bf Symbol & \bf Description \\
\hline
\hline
$G=(V,E)$, $|V|=N$ & Road network $G$ with node set $V$ and edge set $E$\\
$\mathcal{T}$, $|\mathcal{T}|= m$ & Set of trajectories\\
$\mathcal{S}$, $|\mathcal{S}|=n$ & Set of candidate sites \\
\hline
$\ds(u,v)$ & Distance of shortest path from node $u$ to $v$ \\
$\dr(u,v)$ & Round-trip distance between nodes $u$ and $v$ \\
$\dr(T_j,s_i)$ & Round-trip distance from trajectory $T_j$ to site $s_i$\\
\hline
$k$ & Desired number of service locations\\
$\tau$ & Coverage threshold \\
$\pref(T_j,s_i)$ & Preference function for trajectory $T_j$ and site $s_i$\\
\hline
$\mathcal{Q} \subseteq \mathcal{S}$, $|\mathcal{Q}| = k$ & Set of locations to set up service\\
$U_j$ & Utility of trajectory $T_j$ over the set of sites $\mathcal{Q}$ \\
$U(\mathcal{Q})=\sum_{j=1}^m U_j$ & Total utility offered by $\mathcal{Q}$ \\
\hline
\end{tabular}
\tabcaption{Symbols used in the paper.}
\label{tab:symbol}
\moveups
\end{table}

Consider a road network $G = \{V,E\}$ over a geographical area where $V$
denotes the set of road intersections, and $E$ denotes the set of road segments
between two adjacent road intersections.  The direction of the underlying
traffic that passes over a road segment is modeled by directed edges.

Assume a set of candidate sites $\mathcal{S} = \{s_1, \cdots, s_n\} \subseteq
V$ where a certain service or facility can be set up.  The choice of
$\mathcal{S}$ is generally provided by the application itself by taking into
account various factors such as availability, reputation of neighborhood, price
of land, etc.  Most of these factors are outside the purview of the main focus
of this paper and are, therefore, not studied.  We simply assume that the set
$\mathcal{S}$ is provided as input to our problem.  However, as described
later, if all the latent factors of choosing a site can be encapsulated as its
cost, we can handle it quite robustly.

The candidate sites can be located anywhere on the road network.  If it is
already on a road intersection, then it is part of the set of vertices $V$.  If
not, i.e., if it is on the middle of a road connecting vertices $u$ and $v$,
without loss of generality, we augment $V$ to include this site as a new vertex
$w$.  We augment the edge set $E$ by two new edges $(u,w)$ and $(v,w)$ (with
appropriate directions) and remove the old edge $(u,v)$.  Thus, ultimately,
$\mathcal{S} \subseteq V$.

The set of candidate sites $\mathcal{S}$ can be in addition to the set of
existing service locations $\mathcal{E_S}$.

The set of trajectories over the road network is denoted by $\mathcal{T} =
\{T_1, \cdots, T_m\}$ where each trajectory is a sequence of nodes, $T_j =
\{v_{j_1}, \cdots, v_{j_l}\}$, $v_{j_i}\in V$.  The trajectories are usually
recorded as GPS traces and may contain arbitrary spatial points on the road
network. For our purpose, each trajectory is map-matched \cite{map1} to form a
sequence of road intersections through which it passes. We assume that each
trajectory belongs to a separate user. However, the framework can easily
generalize to multiple trajectories belonging to a single user by taking union of each of these trajectories. 

Suppose $\ds(v_i,v_j)$ denotes the shortest network distance along a directed
path from node $v_i$ to $v_j$, and $\dr(v_i,v_j)$ denotes the shortest distance
of a round-trip starting at node $v_i$, visiting $v_j$, and finally returning
to $v_i$, i.e., $\dr(v_i,v_j)=\ds(v_i,v_j)+\ds(v_j,v_i)$.  In general,
$\ds(v_i,v_j) \neq \ds(v_j,v_i)$, but $\dr(v_i,v_j) = \dr(v_j,v_i)$.  With a
slight abuse of notation, assume that $\dr(T_j,s_i)$ denotes the \emph{extra}
distance traveled by the user on trajectory $T_j$ to avail a service at site
$s_i$. Formally, $\dr(T_j,s_i) = \min_{\forall v_k, v_l \in T_j} \{
\ds(v_k,s_i) + \ds(s_i,v_l) - \ds(v_k,v_l) \}$.

It is convenient for a user to avail a service only if its location is not too
far off from her trajectory.  Thus, beyond a distance $\tau$, we assume that the
utility offered by a site $s_i$ to a trajectory $T_j$ is $0$.  We call this
user-specified distance $\tau$ as the \emph{coverage threshold}.

\begin{defn}[Coverage] 
	A candidate site $s_i$ \emph{covers} a trajectory $T_j$ if the distance
	$\dr(T_j,s_i)$ is at most $\tau$, where $\tau \geq 0$ is the \emph{coverage
	threshold}.
\end{defn}

For all sites within the coverage threshold $\tau$, the user also specifies a
preference function $\pref$.  The preference function $\pref(T_j,s_i)$ assigns a
score (normalized to $[0,1]$) for a trajectory $T_j$ and a site $s_i$ that
indicates how much $s_i$ is preferred by the user on trajectory $T_j$.  Higher
values indicate higher preferences with $0$ indicating no preference. In
general, sites that are \emph{closer} to the trajectory have \emph{higher}
preferences than those farther away. Usage of such preference functions are common in location analysis literature \cite{zeng2009pickup}.

\begin{defn}[Preference Function $\pref$]
	\label{def:pref}
$\pref:(\mathcal{T},\mathcal{S})\rightarrow[0,1]$ is a real-valued preference
function defined as follows:
\begin{align}
	\pref(T_j,s_i)=
		\begin{cases}
			f(\dr(T_j,s_i)) & \text{ if } \dr(T_j,s_i) \le \tau \\
			0 & \text{otherwise}
		\end{cases}
\end{align}
where $f$ is a non-increasing function of $\dr(T_j,s_i)$. 
\end{defn}
 
\begin{table}[t]\scriptsize
	\centering
	\begin{tabular}{|c|c|c|c|}
	\hline
	\multirow{2}{*}{Trajectories} & \multicolumn{3}{|c|}{Preference scores for different sites} \\
	\cline{2-4}
		& $s_1$ & $s_2$ & $s_3$\\
			\hline 
            $T_1$ & $\pref(T_1,s_1)=0.4$ & $\pref(T_1,s_2)=0.11$ & $\pref(T_1,s_3)=0$\\
       \hline
    $T_2$ & $\pref(T_2,s_1)=0$ & $\pref(T_2,s_2)=0.5$ & $\pref(T_2,s_3)=0.6$ \\
       \hline   
		\end{tabular}
		\tabcaption{Examples of trajectories with site preferences.}
		\label{tab:example}
		\moveup
\end{table}

The goal of \tops query is to report a set of $k$ sites $\mathcal{Q}
\subseteq \mathcal{S},\ |\mathcal{Q}| = k$, that maximizes the preference score
over the set of trajectories. The preference score of a trajectory $T_j$ over a
set of sites $\mathcal{Q}$ is defined as the \emph{utility function} $U_j$ for
$T_j$, which is simply the \emph{maximum} score corresponding to the sites in
$\mathcal{Q}$, i.e., $U_j = \max_{s_i \in \mathcal{Q}}\{\pref(T_j,s_i)\}$. The various symbols used in the \tops formulation are listed in
Table~\ref{tab:symbol}.

The generic \tops query formulation is stated next.

\begin{prob}[\tops]
	Given a set of trajectories $\mathcal{T}$, a set of candidate sites
	$\mathcal{S}$ that can host the services, \tops problem with query
	parameters $(k,\tau, \pref)$ seeks to report the best $k$ sites,
	$\mathcal{Q} \subseteq \mathcal{S}, \ |\mathcal{Q}| = k$, that maximize the
	sum of trajectory utilities, i.e., $\mathcal{Q} = \arg\max \sum_{j=1}^m U_j$
	where $U_j = \max_{s_i \in Q}\{\pref(T_j,s_i)\}$.
\end{prob}

\textbf{Extensions and variants of \tops:} The \tops framework as defined above
is \emph{generic}.  It can have various extensions as described in Sec.~\ref{sec:cost} and Sec.~\ref{sec:capacity} and can work with existing facilities as well (Sec.~\ref{sec:existing}).  Finally, the preference function in Def.~\ref{def:pref} subsumes  various existing models as exemplified in Sec.~\ref{sec:other}.

\subsection{Properties of \tops}

\tops problem is NP-hard.  To show that, we first define a specific instance
of \tops where the preference function $\pref(T_j,s_i)$ is a binary function.  A
binary function is a natural choice in situations where a service provider wants
to intercept the maximum number of trajectories that pass within a vicinity of
it \cite{berman1995locatingMain}.

\begin{defn}[Binary Instance of \tops]
	\label{def:binary}
	\tops with parameters $(k,\tau,\pref)$ where the preference score is given by:

	\begin{align}
		\pref(T_j,s_i)=
			\begin{cases}
				1 & \text{ if } \dr(T_j,s_i) \le \tau \\
				0 & \text{otherwise}
			\end{cases}
	\end{align}
\end{defn}

\begin{thm}
	\label{thm:nphard}
	\tops is NP-hard.
\end{thm}

\begin{proof}
	This follows from the fact that the binary instance of \tops is NP-hard, as set cover problem is reducible to it \cite{berman1995locatingMain}.
%
	\hfill{}
\end{proof}

Next we show that the sum of utilities $U = \sum_{j=1}^m U_j$ is a
\emph{non-decreasing sub-modular} function.
A function $f$ defined on any
subset of a set $\mathcal{S}$ is \emph{sub-modular} if for any pair of subsets
$\mathcal{Q,R} \subseteq \mathcal{S}$, $f(\mathcal{Q}) + f(\mathcal{R}) \geq
f(\mathcal{Q} \cup \mathcal{R}) + f(\mathcal{Q} \cap \mathcal{R})$
\cite{nemhauser1978analysis}.

\begin{thm}
	\label{thm:submodular}
	%
	$U = \sum_{j=1}^m U_j$ is non-decreasing and sub-modular. 
\end{thm}

\begin{proof}
	Consider any two pair of subsets $\mathcal{Q,R} \subseteq \mathcal{S}$ such
	that $\mathcal{Q} \subseteq \mathcal{R}$. Since $U_j = \max_{s_i \in
	\mathcal{Q}}\{\pref(T_j,s_i)\}$, $U_j(\mathcal{R})\geq U_j(\mathcal{Q})$.
	Therefore, $\sum_{j=1}^m U_j(\mathcal{R})=U(\mathcal{R})\geq \sum_{j=1}^m
	U_j(\mathcal{Q})=U(\mathcal{Q})$. Hence, $U$ is a non-decreasing function.

	To show that $U$ is sub-modular, following \cite{nemhauser1978analysis}, it
	is sufficient to show that for any two subsets $\mathcal{Q,R} \subseteq
	\mathcal{S}, \ \mathcal{Q} \subset \mathcal{R}$, for any site $s \in
	\mathcal{S} \setminus \mathcal{R}$, $U(\mathcal{Q}\cup \{s\})-U(\mathcal{Q})
	\geq U(\mathcal{R}\cup \{s\})-U(\mathcal{R})$.  For our formulation, it is,
	thus, enough if for any trajectory $T_j \in \mathcal{T}$,
	\begin{align}
		\label{eq:submodular}
		U_j(\mathcal{Q}\cup \{s\})-U_j(\mathcal{Q})
			\geq U_j(\mathcal{R}\cup \{s\})-U_j(\mathcal{R})
	\end{align}

	Suppose the site $s^* \in \mathcal{R}\cup \{s\}$ offers the maximal utility
	to trajectory $T_j$ in the set of sites $\mathcal{R}\cup \{s\}$.  There can
	be two cases: \\
	(a) $s^* = s$: $U_j(\mathcal{Q} \cup \{s\})=U_j(\{s\})=U_j(\mathcal{R} \cup
	\{s\})$. Further, $U_j(\mathcal{Q})\leq U_j(\mathcal{R})$ since
	$U_j=\{\max\{\pref(T_j,s_i)\}|s_i \in Q\}$.  These two inequalities together lead to
	Ineq.~\eqref{eq:submodular}. \\
	(b) $s^* \neq s$: Here, $U_j(\mathcal{R} \cup
	\{s\})-U_j(\mathcal{R}) = 0$.  As $U_j()$ is non-decreasing,
	$U_j(\mathcal{Q} \cup \{s\})  \geq  U_j(\mathcal{Q})$.  Thus,
	Ineq.~\eqref{eq:submodular} follows.
	%
	%
	\hfill{}
\end{proof}



\section{Algorithms for \tops}
\label{sec:algo}

\begin{table}[t]\scriptsize
	\centering
		\begin{tabular}{|c|c|c|}
			\hline
			Algorithm & Selected sites $\mathcal{Q}$ & Utility $U$ \\
			\hline
			\textsc{Optimal} &  $\{s_1,s_3\}$ & $1.0$ \\
			\incg   &  $\{s_1,s_2\}$ & $0.9$ \\
			\hline
		\end{tabular}
		\tabcaption{Utilities of different algorithms for
		Example~\ref{ex:one} with $k = 2$.}
		\label{tab:utilities}
		\moveup
\end{table}

We first state an example, using which we will explain the different algorithms
for \tops problem.

\begin{example}
	\label{ex:one}
	\emph{Consider Table~\ref{tab:example}, which lists $2$ trajectories
	$T_1,T_2$ with their preference scores for $3$ sites $s_1$, $s_2$ and $s_3$.
	We set $k=2$ sites to be selected for \tops query. The utilities
	achieved by different algorithms (discussed next) are shown in
	Table~\ref{tab:utilities}. The optimal solution consists of the sites $s_1$ and
$s_3$ with a total utility of $1.0$.}
\end{example}

\subsection{Optimal Algorithm}

We first discuss the \emph{optimal} solution to \tops problem in
the form of an \emph{integer programming problem} (ILP):
\begin{align}
	\text{Maximize } U &= \sum_{j=1}^m U_j \label{eq:obj} \\
	\text{such that } \sum_{i=1}^n x_i &\leq k, \label{eq:cardinality} \\
	\forall j=1,\dots,m, \ 
	U_j & \leq \max_{1\le i \le n} \{\pref(T_j,s_i)\times x_i\}
	\label{eq:MaxUtilityConstraint} \\
	\forall i = 1, \dots, n, \ 
	x_i &\in \{0,1\}
	\label{eq:binary}	
\end{align}

The above IP is, however, not in the form of a standard integer linear program
(ILP).  To do so, the constraints in Ineq.~\eqref{eq:MaxUtilityConstraint} can
be linearized, as discussed in Appendix~\ref{app:ip}. Since \tops is NP Hard, the cost of solving this optimal algorithm is \emph{exponential} with respect to the
number of trajectories $m$, and the number of sites $n$ and, therefore,
\emph{impractical} except for very small $m$ and $n$.  This is demonstrated through experiments in Sec.~\ref{sec:comparison with optimal}.

\subsection{Approximation Algorithms}
\label{sec:approx}

We next present a greedy approximation
algorithm. Before that, let us define a few
terms and how they are computed.



Given a site $s_i \in \mathcal{S}$, $\tc(s_i)$ denotes the set of trajectories
covered by $s_i$, i.e., $\tc(s_i)=\{T_j|\dr(T_j,s_i) \leq \tau\}$.  Similarly,
given a trajectory $T_j$, $\stc(T_j)$ denotes the set of sites that cover $T_j$,
i.e., $\stc(T_j)=\{ s_i|\dr(T_j,s_i) \leq \tau\}$.  For each site $s_i$, the
site weight $w_i = \sum_{j=1}^m \{\pref(T_j,s_i)|T_j \in \tc(s_i)\}$ denotes the
sum of preference scores of all the trajectories covered by it.
Table~\ref{tab:netclus} tabulates some of the important notations.

Since $\tau$ is available only at query time, the computation of these covering
sets need to be efficient.  Hence, for each site $s_i \in \mathcal{S}$, the
shortest path distances to all other nodes in $V$ are pre-computed.  This
requires $O(|\mathcal{S}|.|E|\log |V|)$ time \cite{algorithmsbook}.  As the road
network graphs are almost planar, i.e., $|E|=O(|V|)$, the above cost simplifies
to $O(n.N \log N)$ where $n = |\mathcal{S}|$ and $N = |V|$.

Using these distance values, the distance between each pair of site and
trajectory, $\dr(T_j,s_i)$, is computed.  Suppose there are at most $l$ nodes in
a trajectory $T_j$, i.e., $\max_{T_j \in \mathcal{T}}|T_j| = l$. Referring to
the definition of $\dr(T_j,s_i)$ in Sec.~\ref{sec:formulation},  it follows that
$\dr(T_j,s_i)$ can be computed  in $O(l^2)$ time for a given site-trajectory
pair using the above site-node distances.  Therefore, the total cost across all
$n$ sites and $m$ trajectories is $O(mnl^2)$.

Each site $s_i$ maintains the set of trajectories in an ascending order of
$\dr(T_j,s_i)$ while each trajectory $T_j$ maintains the sites in a similar
ascending order based on $\dr(T_j,s_i)$ again.  Sorting increases the time to
$O(mn(\log n + \log m))$.  The space complexity for storing all site to
trajectory distances is $O(mn)$.

At query time, when $\tau$ is available, the sets $\tc$, $\stc$, and the site
weights can be, thus, computed from the stored distance matrix in $O(mn)$ time.

Alternatively, the sets \tc and \stc can be computed at query time when $\tau$
is made available.  However, for a city-sized road network, even a small $\tau$
includes a lot of neighbors (e.g., for the Beijing dataset described in
Sec.~\ref{sec:exp}, there are more than $40$ sites on average within a distance
of $\tau = 0.8$ Km from a site).  Since this needs to be computed for all the
candidate sites, this approach is infeasible at query time.

\subsection{Inc-Greedy}

We adapt the general greedy heuristic for maximizing non-decreasing
sub-modular functions \cite{nemhauser1978analysis} to design \incg. The main
idea is based on the principle
of \emph{maximizing marginal gain}.

Then the algorithm starts with an empty set of sites $\mathcal{Q}_0=\varnothing$, and
incrementally adds the sites (one in each of the $k$ iterations) such that each
successive addition of a site produces the maximal marginal gain in the utility
$U$.  More specifically, if the set of selected sites after $\theta-1$
iterations $(\theta=1,\dots,k)$ is
$\mathcal{Q}_{\theta-1}=\{s_1,\dots,s_{\theta-1}\}$, in the $\theta^{th}$
iteration, \incg chooses the site $s_\theta \notin \mathcal{Q}_{\theta-1}$ such
that $U(\mathcal{Q}_{\theta} = \mathcal{Q}_{\theta-1}\cup \{s_{\theta}\}) -
U(\mathcal{Q}_{\theta-1})$ is \emph{maximal}.

The marginal utility gained due to the addition of $s_i$ to the set
$\mathcal{Q}_{\theta-1}$ in iteration $\theta=1,\dots,k$ is $U_{\theta}(s_i) =
U(\mathcal{Q}_{\theta-1} \cup \{s_i\}) - U(\mathcal{Q}_{\theta-1})$.  Thus,
$U(\mathcal{Q}_\theta)=\sum_{i=1}^\theta U_i(s_i)$.  
If $\alpha_{ji}$ denotes
the marginal gain in the utility of trajectory $T_j$ by addition of site $s_i$
to the existing set of chosen sites, then, at the end of each iteration $\theta
= 1, \dots, k$, $U_\theta(s_i) = \sum_{\forall T_j \in \tc(s_i)} \alpha_{ji}$.

\incg algorithm begins by computing the sets $\tc,\stc$, and the
site-weights as explained above. It then initializes the
marginal utilities of each site $s_i \in \mathcal{S}$ to its site weight, i.e.,
$U_1(s_i) = w_i$.  Further, for each trajectory $T_j \in \tc(s_i)$, it
initializes $\alpha_{ji} = \pref(T_j,s_i)$.

In iteration $\theta=1,\dots,k$, it selects the site $s_{\theta}$ with
\emph{maximal} marginal utility, i.e., $\forall s_i \notin
\mathcal{Q}_{\theta-1}, U_\theta(s_\theta) \geq U_\theta(s_i)$.  If multiple
candidate sites yield the same maximal marginal gain, it chooses the one with
maximal weight. Still, if there are ties, then without loss of generality, the
site with the highest index is selected.

For each trajectory $T_j \in \tc(s_\theta)$, first the utility $U_j$ is updated
as follows: $U_j \leftarrow \max(U_j,\pref(T_j,s_\theta))$.  In case there is a
gain in $U_j$, the marginal utility of each site $s_i \in \stc(T_j)$ is updated. 

The pseudocode of \incg is prsented in Appendix~\ref{app:inc-greedy}.

In Example~\ref{ex:one}, at iteration $1$, the site $s_2$ has the largest
marginal utility of $0.61$ and is, therefore, selected.  With
$\mathcal{Q}_1=\{s_2\}$, the marginal utility of adding the site $s_1$ is
$0.4-0.11=0.29$, and that of site $s_3$ is $0.6-0.5=0.1$.  Hence, \incg selects
the sites $s_1$ and $s_2$, and yields an utility $U=0.9$, as indicated in
Table~\ref{tab:utilities}.


We next analyze the quality of \incg.

\begin{lem}
	%
	\label{lem:inc_topso2}
	%
	The approx. bound of \incg is $(1-1/e)$.
\end{lem}

\begin{proof}
	
	Since \incg is a direct adaptation of the generic greedy algorithm for
	maximizing non-decreasing sub-modular functions under cardinality
	constraints \cite{nemhauser1978analysis} with $U(\varnothing)=0$, it offers the approximation bound of
	$\left( 1-(1-\frac{1}{k})^k \right) \geq \left(1-\frac{1}{e}\right)$.
	%
	\hfill{}
\end{proof}

\begin{lem}
	\label{lem:inc_topso3}
	$U(\mathcal{Q}_k) \geq (k/n). U(\mathcal{S})$.
\end{lem}

\begin{proof}
	From Theorem~\ref{thm:submodular} it follows that successive marginal
	utilities are non-increasing, i.e., $U_\theta(s_\theta) \geq
	U_{\theta+1}(s_{\theta+1})$. Thus,
	\begin{small}
	\begin{align}
	U_\theta(s_\theta) \geq \frac{\sum_{i=\theta}^n U_i(s_i)}{n-\theta+1}
		\geq \frac{U(\mathcal{Q}_n)-U(\mathcal{Q}_{\theta-1})}{n-\theta+1}
		\label{eq:sum of marginal utilities}
	\end{align}
	 \end{small}
	 
	Next, we claim that $\forall \theta=1,2,\dots,n$,
	\begin{small}
	\begin{align*}
	U(\mathcal{Q}_{\theta}) \geq (\theta / n). U(\mathcal{Q}_n)
	\end{align*}
\end{small}
	We will prove this by induction on $\theta$.
	
	Consider the base case $\theta=1$.  Using Ineq.~\eqref{eq:sum of marginal
	utilities}, $U(\mathcal{Q}_1) = U_1(s_1) \geq U(\mathcal{Q}_n)/n$ since
	$U(\mathcal{Q}_0)=0$. Thus, the induction hypothesis is true for the base
	case.
	
	Next, we assume it to be true for iteration $\theta$.
	For iteration $\theta+1$,
	\begin{small}
	\begin{align*}
	U(\mathcal{Q}_{\theta+1})&= U(\mathcal{Q}_\theta) + U_{\theta+1}(s_{\theta+1})\\
	&\ge  U(\mathcal{Q}_\theta) + \frac{U(\mathcal{Q}_n)-U(\mathcal{Q}_\theta)}{n-\theta}\\
	&\ge \left(\frac{n-\theta-1}{n-\theta}\right)U(\mathcal{Q}_\theta) + \left(\frac{1}{n-\theta}\right)U(\mathcal{Q}_n)
	\end{align*}
	\end{small}

	Using the induction hypothesis for iteration $\theta$,
	\begin{small}
	\begin{align*}
	U(\mathcal{Q}_{\theta+1})&\ge \left(\frac{n-\theta-1}{n-\theta}\right)\left(\frac{\theta}{n}\right) U(\mathcal{Q}_n) + \left(\frac{1}{n-\theta}\right)U(\mathcal{Q}_n)\\
	&\ge \left(\frac{\theta+1}{n}\right). U(\mathcal{Q}_n) 
	\end{align*}
	\end{small}
	
	Thus, the induction hypothesis holds true for any $\theta=1,\dots,n$.
	
	Since $\mathcal{Q}_n=\mathcal{S}$, after $k$ iterations, $U(\mathcal{Q}_k)
	\geq (k/n). U(\mathcal{S})$.	 
	\hfill{}
\end{proof}

\begin{lem}
	\label{lem:inc-greedy topso}
	The approximation bound of \incg is $(k/n)$.
\end{lem}

\begin{proof}
	Since $U$ is non-decreasing (from Th.~\ref{thm:submodular}) and $OPT
	\subseteq \mathcal{S}$, therefore  $U(OPT) \leq U(\mathcal{S})$. The rest follows from Lem.~\ref{lem:inc_topso3}. 
	\hfill{}
\end{proof}

\begin{thm}
	\label{thm:inc_topso}
	The approximation bound of \incg for \tops is $\max\{1-1/e, k/n\}$.
\end{thm}
\begin{proof}
	The result follows from Lem.~\ref{lem:inc_topso2} and
	Lem.~\ref{lem:inc-greedy topso}.
	\hfill{}
\end{proof}

The next result establishes the complexity of \incg.

\begin{thm}
	\label{thm:complexity of incg}
	The time and space complexity bounds for \incg are $O(k.mn)$ and $O(mn)$
	respectively.
\end{thm}
\begin{proof}
	As discussed earlier in Sec.~\ref{sec:approx}, the computation of the sets
	$\tc$, $\stc$ and the weights require $O(mn)$ time and storage.
	
	Initializing the marginal utilities $U_\theta(s_i)$ of all the sites and
	$\alpha_{ji}$ values for all trajectory site pairs require $O(mn)$ time.
	
	In any iteration $\theta=1,\dots,k$, selecting the site with maximal
	marginal utility requires at most $O(n)$ time.  The largest size of any \tc
	set is at most $m$ and that of any \stc set is at most $n$.  Hence, scanning
	the set of trajectories $\tc(s_\theta)$ 
	and updating the utilities requires $O(m)$ time.  If $U_j$ gets updated,
	then scanning the set of sites in $\stc(T_j)$
	and updating the values of $U_\theta(s_i)$ and $\alpha_{ji}$ requires $O(n)$
	time.  Thus, the time complexity of any iteration $\theta$ is $O(mn)$.
	Hence, the total time complexity of $k$ iterations is $O(kmn)$.

	Since no other space is required than that for storing \tc, \stc and site
	weights, the total space complexity is $O(mn)$.
	\hfill{}
\end{proof}

\subsection{Limitations of Inc-Greedy}
\label{sec:limitations}

Although \incg provides a constant factor approximation, it is not scalable to
large real-life datasets.  The reasons are:

\noindent
\textbf{High query cost}: 
	The input parameters for \tops query,  $(k,\tau, \pref)$, are available only at query time.  Hence, the
	covering sets $\tc$, $\stc$, and the site weights (that depend on the value of $\tau$)  can be generated \emph{only}
	at query time.  Even if all pairwise site-to-trajectory distances are pre-computed, this step requires a  high computation cost of $O(mn)$ (both in terms of time
	and memory) where $m,n$ denote the number of trajectories, and candidate sites, respectively. Due to this reason, for real city-scale datasets
	(such as the Beijing dataset \cite{cab1} used in our experiments that has
	over 120,000 trajectories and 250,000 sites),  \incg is not scalable. Fig.~\ref{subfig:tautime} shows that \incg takes about 2000 sec. to complete for $\tau=1.2$ Km. and $k=5$, and goes out of memory for $\tau > 1.2$ Km.
	
\noindent
\textbf{High storage cost:}
	As discussed above, to facilitate faster computation of covering sets \tc, we need to pre-compute all pairwise site-to-trajectory distances. 		
	However, for any city-scale dataset, this storage requirement is
	prohibitively large.  For example, the Beijing dataset \cite{cab1} 
	 require close to 250 GB  of
	storage.  This is unlikely to fit in the main memory and, therefore,
	multiple expensive random disk seeks are required at run-time. Even with pre-computed distances up to 10 Km., \incg crashes beyond $\tau \> 1.2$ Km. (shown in Tab.~\ref{tab:memory}.)

\noindent
\textbf{High update cost:}
	\incg is also not amenable to updates in trajectories and sites.  If a
	new trajectory is added, its distance to all the sites needs to be computed
	and sorted.  In addition, the sorted \tc set of all the sites need to be
	updated as well.  Similarly, if a new candidate site is added, the distances
	of all the trajectories to this site will need to be computed and sorted.
Such
	costly update operations are impractical, especially at run-time.

A careful analysis of \incg reveals that there are two main stages of the
algorithm.  In the first, the sets \tc, \stc and the site weights are computed.
In the second stage, some of these sets are updated in an iterative manner.

The first stage is heavier in terms of time and space requirements.  Thus, to
make it efficient, we use an index structure, \nc, which is described in
Sec.~\ref{sec:offline} and Sec.~\ref{sec:online}.  The use of indexing reduces
the computational burden of the update (i.e., the second) stage as well.
Further, Section~\ref{sec:updates} shows how \nc allows easier handling of
additions and deletions of candidate sites and trajectories.

In addition, if the preference function $\pref$ is binary as defined in
Def.~\ref{def:binary}, the update steps for \incg can be performed quite
efficiently using FM sketches \cite{fm,krep}.  We next describe the details.

\subsection{Using FM Sketch to speed up Inc-Greedy}
\label{sec:fm}

The main use of FM sketches is in counting the number of \emph{distinct}
elements in a set or union of sets \cite{fm}.  Suppose, the maximum number of
distinct elements is $N$.  An FM sketch is a bit vector, which is initially
empty, and is of size at least $O(\log_2 N)$.  The probability of an element
from the domain hashing into the $i^\text{th}$ bit of the FM sketch is $2^{-i}$.
Thus, if a set has $\Gamma$ elements, the probability of the last bit marked in
the FM sketch is $O(\log_2 \Gamma)$.  Hence, after the elements of a set are
hashed to the FM sketch, the last set bit can be used to estimate the number of
distinct elements in the set.  (The details of how the hashing function is
chosen and the exact estimates are in \cite{fm}.) Although the FM sketch does
not count the number of distinct elements exactly, it provides a multiplicative
guarantee on the error in counting.  When more copies of the FM sketch is used,
the error decreases.

The FM sketch can be used to speed up the update stage of \incg, since selecting
a site with the largest maximal utility is the same as selecting a site that
covers the largest number of \emph{distinct} trajectories not yet covered.

For each site $s_i \in \s$, the set of trajectories that it covers, i.e.,
$\tc(s_i)$ is maintained as an FM sketch.  Thus, instead of maintaining
$O(m)$-sized lists for each site where $m$ is the total number of trajectories,
we only need to maintain $O(\log_2 m)$-sized bit vectors per site.

Suppose the count of distinct trajectories covered by a site $s_i$ is $\chi_i$.
The \emph{marginal} utility of site $s_j$ when site $s_i$ has been chosen is the
number of distinct trajectories that the two sites \emph{together} cover over
the number of trajectories that site $s_i$ \emph{alone} covers.  The estimate of
the number of trajectories covered by the union of $s_j$ and $s_i$ can be
obtained by the bitwise OR of the FM sketches corresponding to $s_j$ and $s_i$.
If this estimate is $\chi_{ij}$, the marginal utility of site $s_j$ over site
$s_i$ is $\chi_{ij} - \chi_i$.

Therefore, when there are $n$ candidate sites, to determine the site that
provides the best marginal utility over site $s_i$, $n-1$ such bitwise OR
operations are performed, and the maximum is chosen.  At the end of the
$\theta^\text{th}$ iteration, the combined number of trajectories covered by the
sites in $\mathcal{Q}_{\theta}$ is stored by the \emph{union} of the FM sketches
obtained successively in the $\theta$ iterations.  The $(\theta+1)^\text{th}$
site is chosen by using this combined FM sketch as the base.

The above brute-force algorithm can be improved in the following way.  The upper
bound of the marginal utility for any site $s_j$ is its own utility.  Thus, if
the current best marginal utility of another site $s_k$ is already greater than
that, it is not required to do the union operation with $s_j$.  If the sites are
\emph{sorted} according to their utilities, the scan can stop as soon as the
first such site $s_j$ is encountered.  All sites having a lower utility are
guaranteed to be useless as well.

In our implementation, the FM sketches are stored as $32$-bit words.  This
allows handling of roughly $2^{32}$ (which is more than $4$ billion) number of
trajectories.  The length $32$ is chosen since the bitwise OR operation of two
such regular-sized words is extremely fast in modern operating systems. 

\section{Offline Construction of NetClus}
\label{sec:offline}

As discussed in Sec.~\ref{sec:limitations}, \incg has two computationally
expensive components. While FM sketch expedites the information update
component, computation of \tc, \stc, etc. still remain a bottleneck with
$O(mn)$ time and storage complexity. To overcome this scalability issue, we
develop an index structure.  

One of the most natural ways to achieve the above objectives is to cluster the
sites in the road network to reduce the number, and then apply \incg on the
cluster representatives.  The clustering of sites can be done according to two
broad strategies.  The first is to cluster the sites that are highly similar in terms of their trajectory covering sets \tc. The similarity between these covering sets can be quantified using  the \emph{Jaccard similarity} measure (Appendix~\ref{app:Jaccard}).
 However, this
approach is not practical due to two major limitations: (1)~Since the coverage threshold $\tau$ is available \emph{only} at query time, the covering sets \tc can be computed only at query time. Hence, the clustering can be performed only at query time. This leads to impractical query time. 
(2)~Alternatively, multi-resolution clustering may be performed based on few
fixed values of $\tau$, and at query time, a clustering instance of a
particular resolution is chosen based on the value of the query parameter
$\tau$.  However, this still requires computing the similarity between each pair of sites. Owing to large number of sites, and large size of the covering sets, such computation demands impractical memory overhead. 
Hence, we adopt the second clustering option, that of
distance-based clustering. 
  
We first state one basic observation.  If two sites are close, the sets of
trajectories they cover are likely to have a high overlap.  Hence, when $k \ll
n$, which is typically the case, the sites chosen in the answer set are likely
to be distant from each other.  The index structure, \nc, is designed based on
the above observation.
  
Our method follows two main phases: offline and online.  In the offline phase,
clusters are built at multiple resolutions.  This forms the different index
instances.  A particular index instance is useful for a particular range of
query coverage thresholds.  In the online phase, when the query parameters are
known, first the appropriate index instance is chosen.  Then the \incg algorithm is
 run with the cluster representatives of that instance.
 
We explain the
offline phase in this section and the online phase in the next. The important
notations used in the \nc scheme are listed in Table~\ref{tab:netclus}. 

\subsection{Distance-Based Clustering}
\label{sec:clustering}

The clustering method is parameterized by a distance threshold $R$, which is
the \emph{maximum cluster radius}.  The round-trip distance from any node
within the cluster to the cluster-center is constrained to be at most $2R$.
The radius is varied to obtain clusters at multiple resolutions.  We describe
the significance of the choices of $R$ later.

The objective of the clustering algorithm is to partition the set of nodes in
the road network, $V$, into disjoint clusters such that the number of such
clusters is \emph{minimal}.  This leads not only to savings in index storage,
but more importantly, it results in faster query time, as \incg is run on a
smaller number of cluster representatives.  We next describe how to achieve this
objective.

\subsubsection{Generalized Dominating Set Problem (GDSP)}

Given an undirected graph $G=(V,E)$, the \emph{dominating set problem (DSP)}
\cite{haynes1998fundamentals} computes a set $D \subseteq V$ of \emph{minimal}
cardinality, such that for each vertex $v \in V - D$, there exists a vertex $u
\in D$, such that $(u,v) \in E$.  DSP is NP-hard \cite{haynes1998fundamentals}.
In \cite{chen2012approximation}, it was generalized to the \emph{measured
dominating set} problem for weighted graphs.  In this work, we propose a
\emph{generalized dominating set problem (GDSP)}, that uses a different notion
of dominance from \cite{chen2012approximation}.

\begin{prob}[GDSP]
	Given a weighted directed graph $G=(V,E,W)$ where $W: E \rightarrow
	\mathbb{N}$ assigns a positive weight for each edge in $E$, and a constant
	$R>0$, a vertex $u \in V$ is said to \emph{dominate} another vertex $v \in
	V$ if $d(u,v) + d(v,u) \leq 2R$, where $d$ denotes the directed path weight.
	The \emph{GDSP} problem computes a set $D \subseteq V$ of \emph{minimal}
	cardinality such that for any $v \in V - D$, there exists a vertex $u \in D$
	such that $u$ \emph{dominates} $v$.
\end{prob}
 
\emph{GDSP} is NP-hard due to a direct reduction from DSP where all edge weights
are assumed to be $1$, and $R = 2$.

\subsubsection{Greedy Algorithm for GDSP}
\label{sec:greedy algo for GDSP}

To solve GDSP, we adapt the greedy algorithm proposed in
\cite{chen2012approximation}.  We refer to our algorithm as \emph{Greedy-GDSP}.
The only input parameter to the clustering process is the cluster radius $R$.


First, the \emph{dominating sets} of every vertex $v$, denoted by $\Lambda(v)$,
are computed.  This is achieved by running the shortest path algorithm from a
source vertex till distance $2R$.
The dominance relationship is \emph{symmetric}, i.e., $u \in \Lambda(v)
\Leftrightarrow v \in \Lambda(u)$.

The main part of the Greedy-GDSP algorithm is iterative.  In the first
iteration, the vertex $v$ that dominates the largest number of vertices is
chosen.  The set of dominated vertices, $\Lambda(v)$, forms a new \emph{cluster}
with $v$ as the cluster center.  The vertices $v$ and $\Lambda(v)$ are not
considered for further comparisons.  In addition, the vertices in $\Lambda(v)$
are removed from the dominating sets of other non-clustered vertices.  In other
words, for each $u \in \Lambda(v)$, if $u \in \Lambda(w)$ for some non-clustered
vertex $w$ then $\Lambda(w) = \Lambda(w) - \{u\}$.  In the subsequent
iterations, the vertex with the largest \emph{incremental} dominating set as
produced from previous iterations is chosen.  The dominated vertices that are
not part of other clusters form a new cluster.  The algorithm terminates when
all the vertices are clustered.

Similar to the \incg algorithm, we use FM sketches to efficiently update the
dominating sets and choose the vertex with the largest incremental dominating
set in each iteration.  The details are  same as those described in
Sec.~\ref{sec:fm} with trajectory covering sets for each candidate site in \incg
replaced by dominating sets for each vertex in Greedy-GDSP.


\subsubsection{Analysis of Greedy-GDSP}

The next two theorems characterize the
approximation bound and time complexity of Greedy-GDSP.

\begin{thm}
	The cardinality of the dominating set computed using the proposed algorithm
	is within an approximation bound of $(1 + \epsilon').(1 + \ln n)$ of the
	optimal where $\epsilon'$ is the approximation error of FM sketches.
\end{thm}
\begin{proof}
	Following the analysis in \cite{chen2012approximation}, the Greedy-GDSP
	algorithm offers an approximation bound of $H_n = 1 + \frac{1}{2} + \dots +
	\frac{1}{n} \leq 1 + \ln n$ which was shown to be tight unless $P=NP$.
	This, however, does not consider the approximation due to the use of FM
	sketches.  Incorporating that, the bound becomes $(1 + \epsilon').(1 + \ln
	n)$ where $\epsilon'$ is the approximation error of FM sketches.
	\hfill{}
\end{proof}

\begin{thm}
	\label{thm:gdsp-complexity}
	Greedy-GDSP runs in $O(|V| . (\nu \log \nu + \eta))$ time where $\nu$ is the
	maximum number of vertices that are reachable within the largest round-trip
	distance $R_{max}$ from any vertex $v$, and $\eta$ is the number of clusters
	returned by the algorithm.
\end{thm}
The proof is given in Appendix~\ref{app:gdsp}.

\begin{table}[t]
\scriptsize
\centering
\begin{tabular}{cl}
\hline
\textbf{Symbol} & \multicolumn{1}{l}{\textbf{Description}} \\
\hline
\hline
$\tc(s)$ & Set of trajectories covered by site $s$ \\
$\cl(g)$ & Neighbors of cluster $g$ \\
$\tl(g)$ & Set of trajectories passing through cluster $g$ \\
$\tc(g)$ & Set of trajectories passing through $g$ and its neighbors $\cl(g)$\\
$\widehat{\dr}(T,s)$ & Estimate of $\dr(T,s)$ in the clustered space\\
$\widehat{\tc}(s)$ & Estimate of $\tc(s)$ in the clustered space\\
$\eps$ & Resolution at which index instances change \\
$t$ & Number of index instances \\
$\inst_p$ & Index instance \\
$R_p$ & Cluster radius for $\inst_p$ \\
$\eta_p$ & Number of clusters for $\inst_p$ \\
$\Lambda(v)$ & Dominating set for node $v$ \\
$f$ & Number of bit vectors for FM sketch \\
$\epsilon$ & Error parameter for FM sketch \\
\hline
\end{tabular}
\tabcaption{Important notations used in the algorithms.}
\label{tab:netclus}
\vspace*{2mm}
\end{table}
\subsection{Selection of Cluster Representatives} 
\label{sec:representative}

In order to run \incg on the clusters, each cluster needs to choose a
\emph{representative} candidate site.  This may be different from the cluster
center that was used to construct the cluster.  The flexibility is needed since
the cluster representative should \emph{necessarily} be a candidate site,
although the cluster center may be any vertex in $V$.  Taking
into account the fact that the cluster representative should summarize the
information about the cluster and the trajectories that pass through it, and
use this information to compete against the other cluster representatives in
the online phase, we study two alternatives of choosing the cluster
representative:

\begin{enumerate}

	\item The \emph{most frequently} accessed candidate site, i.e., the one
		through which the largest number of trajectories pass through. 

	\item The candidate site that is \emph{closest} to the cluster center. 

\end{enumerate}

While the first option guarantees that the utility of the cluster is at least
that of its best site, the second summarizes the distribution of trajectories
better.
Empirical studies show that the utilities returned by the two alternatives are
quite similar, but the second alternative is marginally better.
Consequently, 
we adopt the second option.

\subsection{Cluster Information}
\label{sec:information}

Suppose the above clustering algorithm produces clusters of radius $R$, i.e., the maximum round-trip distance of any node within a cluster to its cluster center is at most $2R$. Then, a pair of clusters are considered as \emph{neighbors} of each other if their
centers are within a round-trip distance of $4.R(1+\eps)$, where $\eps \in
(0,1]$ is the index resolution parameter to be described in
Sec.~\ref{sec:index}. This choice of neighborhood is explained in Sec.~\ref{sec:topsc problem}.
 
As part of the index structure, every cluster $g_i$ stores the following
information:

\begin{enumerate}

	\item Cluster \emph{center}, $c_i$.

	\item Cluster \emph{representative}, $r_i$.

	\item Trajectory set, i.e., list of \emph{trajectories} passing through at
		least one site in $g_i$, along with their round-trip distance to $c_i$,
		$\tl(g_i) = \langle T_j, \dr(T_j,c_i) \rangle$.

	\item Cluster \emph{neighbors} along with the round-trip distance between
		their centers and $c_i$, $\cl(g_i) = \{\langle g_j, \dr(c_i,c_j)
		\rangle\}$, sorted by $\dr(c_i,c_j)$.

	\item Set of \emph{nodes} in the cluster and their round-trip distance to
		$c_i$, $\{\langle v_i, \dr(v_i,c_i) \rangle\}$.

\end{enumerate}


The set $\tl(g_i)$ is computed by scanning the sequence of nodes in each
trajectory $T_j$.  If $v_i \in T_j$, then $T_j$ is added to $\tl(g(v_i))$,
where $g(v_i)$ denotes the cluster in which $v_i$ resides.  Thus, a trajectory
is represented as a \emph{sequence of clusters}.  As neighboring nodes in any
trajectory are likely to fall into the same cluster, this allows a
\emph{compressed} representation of the trajectory by collapsing consecutive
copies of the same cluster into one.  This compression contributes towards the
efficiency of \nc in terms of both space and time.

\begin{figure}[t]
\centering
\subfloat[Example]{
  \includegraphics[width=0.34\columnwidth]{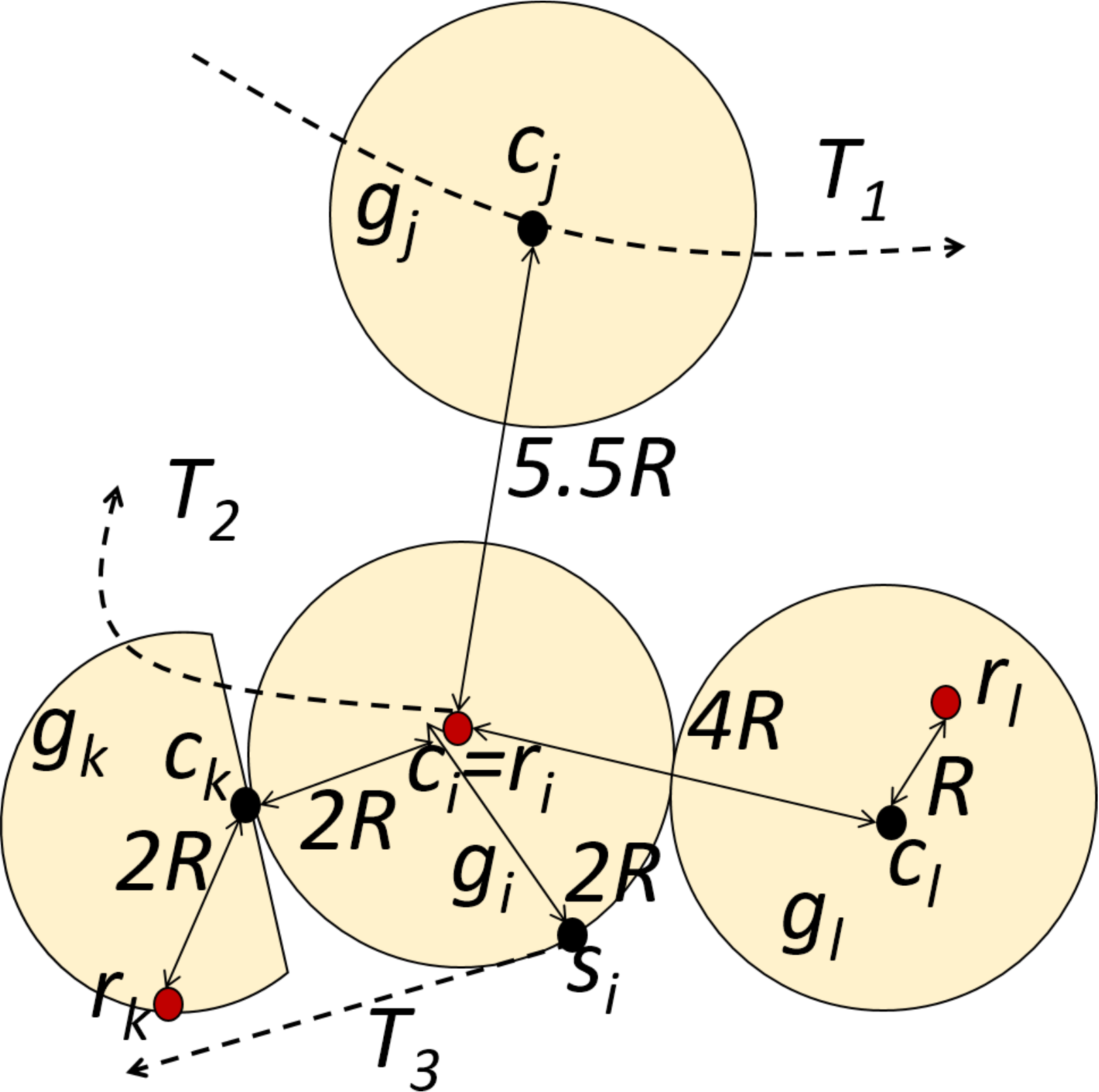}
\label{fig:clusters}
  }
  \subfloat[Covering Sets]{
  \includegraphics[width=0.58\columnwidth]{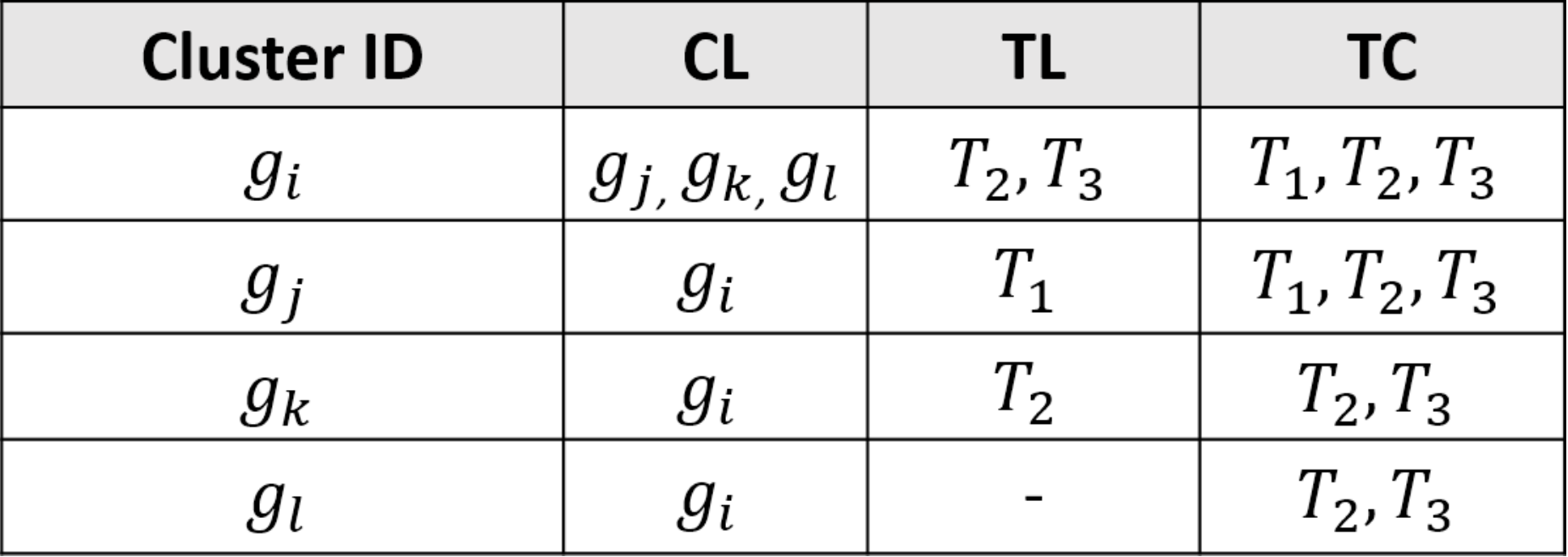}
\label{fig:coverage}
  }
\figcaption{\nc example using $\eps=0.5$.}
\moveup
\end{figure}

\begin{example}
	\emph{
	Fig.~\ref{fig:coverage} illustrates an example of \nc clustering with
	cluster radius $R$ and $\eps=0.5$.
	The clusters $g_i$, $g_j$, $g_k$ and $g_l$ have centers $c_i$, $c_j$, $c_k$ and
	$c_l$ respectively. The cluster $g_i$ has two candidate sites $r_i$ and $s_i$.
	Since $r_i$ is located at $c_i$, it is chosen as the cluster-representative.
	While cluster $g_j$ has no candidate site, each of the clusters $g_k$ and $g_l$
	have one candidate site, namely, $r_k$ and $r_l$, respectively, each of
	which is a
	cluster-representative. The distance between the cluster centers are as
	follows:
	$\dr(c_i,c_j)=5.5R$, $\dr(c_i,c_k) \gtrsim 2R$, and $\dr(c_i,c_l) = 4R$ where
	$\gtrsim$ means just greater
	than. The distance between any other pair of cluster centers is greater
	than or equal to $6R$. Given that $\eps =0.5$, the distance between any two
	neighboring cluster centers lies in the range  $[2R,6R)$. Based on this,
		the cluster neighbors, \cl, are shown in the figure. The trajectories
		$T_1$, $T_2$ and $T_3$ pass through nodes $c_j$, $r_i$ and $s_i$,
		respectively. The figure lists the trajectory sets $\tl$ for each
		cluster. Note that when
		$\tau \ge 4R$, it is guaranteed that any site covers any trajectory
		that passes through the same cluster. For example, $r_k$ covers $T_2$
		as it passes through the cluster $g_k$ that contains $r_k$.  
	}	
	\hfill{}
	$\square$
\end{example}

\subsection{Multi-Resolution Index Structure}
\label{sec:index}

We next explain how the multi-resolution index structure, \nc, is built by using
the clustering algorithm outlined above.  Assume that the normal range of query
coverage threshold $\tau$ is $[\tau_{min}$, $\tau_{max})$.  (We discuss the two
extreme cases later.) \nc maintains $t$ instances of index structures
$\inst_0,\dots,\inst_{t-1}$ of varying cluster radii.  From one instance to the
next, the radius increases by a factor of $\epsinc$ for some $\eps > 0$.  Thus,
the total number of index instances is $t = \lfloor \log_{\epsinc} (\tau_{max} /
\tau_{min}) \rfloor +1$.  For each instance, all the clusters and their
associated information are stored.

Consider a particular index instance $\inst_p$ with cluster radius $R_p$.  As
discussed above, the maximum round-trip distance from a site $s_i$ belonging to
the cluster $g_i$ to a trajectory $T_j$ that passes through $g_i$ is at most
$\dr(T_j,s_i)$ $\leq 4.R_p$.  Thus, if the coverage threshold $\tau < 4.R_p$,
then it is not guaranteed if $s_i$ covers $T_j$ or not.  Hence, the index
instance $\inst_p$ is not useful for any $\tau < 4.R_p$ and a finer instance
with a lesser cluster radius should be used.

On the other hand, if $\tau$ is too large, too many neighboring clusters may
cover a trajectory.  Therefore, intuitively, it makes sense to switch to a
higher index instance with a larger cluster radius so that less number of
clusters need to be processed.  The parameter \epsinc captures the ratio of
$\tau$ to $4.R_p$ beyond which the switch is made.  Thus, if $\tau > 4.R_p .
\epsinc$, a higher index instance is used.

Therefore, the range of useful $\tau$ for the index instance $\inst_p$ is
$[4.R_p, 4.R_p.\epsinc)$.  Hence, the lowest cluster radius is $R_0 = (\tau_{min} / 4)$,
and the successive cluster radii for instances $\inst_p, \ p = 1, \dots, t-1$
are $R_p = \epsinc^{p} R_0$.  From one index instance to the next, as the
cluster radius $R_p$ grows, the number of clusters, $|\inst_p|$, falls
exponentially.


\noindent
\textbf{Choice of $\eps$:} The number of index instances $t$ depends on $\eps$.
A smaller value of $\eps$ creates more number of instances, thereby requiring
larger storage and offline running time.  The approximation error is also
affected by $\eps$.  When $\eps$ is smaller, the range of $\tau$ handled by a
particular index instance is tighter.  Therefore, the distance approximations
are better.
Experimental results showing the empirical impact of \eps are discussed in
Sec.~\ref{sec:parameters}.

\noindent
\textbf{Extreme cases of $\tau$:} The extreme values of the range of $\tau$,
namely, $\tau_{min}$ and $\tau_{max}$, are assigned respectively as the
\emph{minimum} and \emph{maximum} round-trip distance between any two sites in
$\mathcal{S}$.  This particular choice is guided by the following analysis.  If
there is a query with $\tau < \tau_{min}$, then the method degenerates to
normal \incg as each site becomes a cluster by itself.  If, on the other hand,
$\tau \geq \tau_{max}$, then \nc reports any $k$ sites, as each site covers
every other site, and consequently, all the trajectories.  Hence, the
multi-resolution \nc is applicable to \emph{all} query coverage thresholds.

\section{Querying using NetClus}
\label{sec:online}

We next explain the online phase of querying that starts after the query
parameters $(k,\tau)$ are available.

The first important consideration is choosing the index instance $\inst_p$
that supports the given query threshold $\tau$.  The index $p$ is computed as $p
= \lfloor \log_{\epsinc} (\tau / \tau_{min}) \rfloor$.  This ensures that $4.R_p
\leq \tau < 4.R_p. \epsinc$ where $R_p$ is the cluster radius for
$\inst_p$.

We next discuss how to apply \tops on the clustered space.

\subsection{\topsc Problem}
\label{sec:topsc problem}


Consider a cluster $g_i$ with its representative $r_i$, and a trajectory $T_j$ passing through a cluster $g_j$, where $g_j$ may or may not be equal to $g_i$. Then $T_j \in
\tc(r_i)$ if and only if $\dr(T_j,r_i) \leq \tau$.  In the clustered space,
however, we only store the distances of each trajectory from the centers of the clusters that it passes through. Hence, it is not possible to compute $\dr(T_j,r_i)$ without extensive online computation. 
  Hence, an
\emph{approximate} distance $\widehat{\dr}(T_j,r_i)$ is computed and used.  The
round-trip distance estimate from $T_j$ to $r_i$ is
\begin{align}
	\label{eq:approxcover}
	\widehat{\dr}(T_j,r_i) = \dr(T_j,c_j) + \dr(c_j,c_i) + \dr(c_i,r_i)
\end{align}
It is important to note that the distance can be estimated using \emph{only} the
information computed in the offline phase. 
Since the distances are approximate, the \emph{approximate trajectory cover} of
$r_i$ is
\begin{align}
	\label{eq:approx tc}
	\widehat{\tc}(r_i)= \{T_j \in \tc(g_i) | \widehat{\dr}(T_j,r_i) \leq \tau\}
\end{align}
where $\tc(g_i) = \tl(g_i) \cup \{T_j \in \tl(g_j)|g_j \in \cl(g_i)\}$ consists of the
trajectories passing through $g_i$ and  its neighbors $\cl(g_i)$. 

For any $T_j \in \widehat{\tc}(r_i)$,  $T_j \in
\tc(g_i)$, if and only if there exists a cluster $g_j$ such that $T_j \in \tl(g_j)$ and $\dr(c_i,c_j) \leq
\tau$. This follows from the fact that $\widehat{\dr}(T_j,r_i) \leq \tau$.  For the index instance $\inst_p$, since
$\tau \leq 4R_p\epsinc$, therefore, this condition reduces to $\dr(c_i,c_j)
\leq 4R_p\epsinc$. This is the reason why the neighborhood of a cluster is
defined as those whose centers are within a round-trip distance of
$4R_p\epsinc$ in Sec.~\ref{sec:information}.

Consequently, to compute the set $\widehat{\tc}(r_i)$, it is sufficient to
examine only the trajectory sets of the neighbors of the cluster $g_i$.  For
each trajectory $T_j \in \tl(g_j)$ where $g_j$ is a neighbor of $g_i$, the
approximate distance $\widehat{\dr}(T_j,r_i)$ is computed.  The trajectory $T_j$
is included in $\widehat{\tc}(r_i)$ if $\widehat{\dr}(T_j,r_i) \le \tau$.

\begin{table}[t]
	\begin{center}
	\begin{tabular}{c|ccc}
	$\widehat{\dr}$ & $r_i$ & $r_k$ & $r_l$\\
	\hline
	$T_1$ & $5.5R$ & $\gtrsim 9.5R$ & $10.5R$\\
	$T_2$ & $0R$ & $\le 4R$ & $5R$ \\
	$T_3$	& $2R$ & $\gtrsim 6R$ & $7R$\\ 		
		\end{tabular}
		\tabcaption{Distance estimates $\widehat{\dr}(\cdot)$ for example in Fig.~\ref{fig:clusters}.}
		\label{tab:distance-estimates}
\end{center}
\end{table}

\begin{example}
	\emph{
	The $\tc$ sets for the three clusters in Fig.~\ref{fig:clusters} are shown
	in Fig.~\ref{fig:coverage}.  Using Eq.~\eqref{eq:approxcover}, the distance estimates between each pair of cluster representative and trajectory is shown in Table~\ref{tab:distance-estimates}. Since $\eps=0.5$, therefore, the supported range of $\tau$ is
	$[4R, 6R)$. Thus, if $\tau=4R$, $\widehat{\tc}(r_i)=\{T_2,T_3\}$, 
	$\widehat{\tc}(r_k)=\{T_2\}$, and $\widehat{\tc}(r_l)=\varnothing$. Similarly, if $\tau=5.75R$, $\widehat{\tc}(r_i)=\{T_1,T_2,T_3\}$, 
	$\widehat{\tc}(r_k)=\{T_2\}$, and $\widehat{\tc}(r_l)=\{T_2\}$. 
}	
	\hfill{}
	$\square$
\end{example}

If a trajectory $T_j \in \widehat{\tc}(r_i)$, then it also lies in the set
$\tc(r_i)$ since $\dr(T_j,r_i) \leq \widehat{\dr}(T_j,r_i) \le \tau$.  However, the
reverse is not true, since there may be a trajectory $T_j$ such that
$\dr(T_j,r_i) \leq \tau$, but the estimate $\widehat{\dr}(T_j,r_i) > \tau$. Therefore, $\widehat{\tc}(r_i) \subseteq \tc(r_i)$.
For example, in Fig.~\ref{fig:clusters}, $\dr(T_3,r_k) \le 4R$, but $\widehat{\dr}(T_3,r_k) \gtrsim 6R$ which exceeds any supported value of $\tau$. Thus, $T_3 \notin \widehat{\tc}(r_k)$.


Finally, based on the query preference function $\pref$, the preference score
$\pref(T_j,r_i)$ is computed between all cluster representatives $r_i$ and their
trajectory covers $T_j \in \widehat{\tc}(r_i)$.

Using these approximate covering sets, we  run the following instance of \tops
problem, called \textsc{\topsc}.

\begin{prob}[\topsc]
	Given an index instance $\inst_p$ defined over the road network $G =
	(V,E)$, suppose $\widehat{\mathcal{S}} \subseteq S$ denote the set of cluster
	representatives in $\inst_p$.  \topsc problem seeks to report a
	set of $k$ cluster representatives, $\mathcal{Q} \subseteq
	\widehat{\mathcal{S}}$, $|\mathcal{Q}| = k$, such that
	$U(\mathcal{Q})$ is \emph{maximal}.
\end{prob}

 To
solve \topsc, we employ \incg on the set of cluster representatives
$\widehat{\mathcal{S}}$ using the above covering sets $\widehat{\tc}(r_i)$.

When the preference function $\pref$ is binary, FM sketches can be employed for faster updating of marginal utilities during the execution of \incg on the cluster representatives, in the same manner as described in Sec.~\ref{sec:fm}. 

\subsection{Analysis of NetClus}
\label{sec:analysis}
\noindent
\textbf{Quality Analysis:} 
The first result is due to a direct application of Lem.~\ref{lem:inc_topso3}.

\begin{cor}
	\label{cor:approx}
	The utility of the set $\widehat{\mathcal{Q}}$ returned by the \nc framework is
	bounded as follows: $U(\widehat{\mathcal{Q}}) \geq
	(k / |\widehat{\mathcal{S}}|) U(\widehat{\mathcal{S}})$.
\end{cor}

If the index instance $\inst_p$ is used for a particular query threshold
$\tau$, then $|\widehat{\mathcal{S}}|$ is at most the number of clusters, $\eta_p$.
The next result states the approximation guarantees offered by the \nc
framework.

\begin{thm}
	\label{thm:topsbound}
	The approximation bound offered by \nc for the binary instance of \tops is
	$(k / \eta_p)$.  For a general preference function
	$\pref(T_j,s_i)=f(\cdot)$, the approx. bound is $f(\tau). (k / \eta_p)$.
\end{thm}

\begin{proof}
	Assuming all nodes are candidate sites, we observe that each trajectory $T_j \in \traj$ is covered by the cluster representative $r_i$ of a cluster $g_i$ that it passes through. This is because the maximum round-trip distance between $T_j$ and $r_i$ is at most $4R_p$ and $\tau$ is at least
	$4R_p$.
 
	This ensures that $U(\widehat{\mathcal{S}})=\sum_{T_j \in \traj} U_j = |\traj|=m$ because  $\forall T_j \in \traj,\ U_j=1$.
	
	 Since
	the maximum utility of the optimal algorithm for \tops can be at most $m$,
	following the result in Cor.~\ref{cor:approx}, the approximation bound is at
	least $k / \eta_p$ where $\eta_p=|\widehat{\s}|$. 

	Next, consider a general preference function $\pref(T_j,s_i)=f(\cdot)$ where
	$f$ is a positive non-increasing function of $\dr(T_j,s_i)$ such that
	$f(0)=1$.  If $T_j$ passes through the cluster $g_i$, then $\dr(T_j,r_i)
	\leq \tau$. Hence, $U_j = \max\{\pref(T_j,r_i)| r_i \in \widehat{\mathcal{S}}\}
	\geq f(\tau)$ since $f$ is non-increasing.  Therefore,
	$U(\widehat{\mathcal{S}})=\sum_{T_j \in \traj} U_j \geq f(\tau).m$.  Since the
	preference scores lie in the range $[0,1]$, the utility offered by any
	optimal algorithm for \tops is at most $m$.  Therefore, following
	Cor.~\ref{cor:approx}, the approximation bound is at least $f(\tau). (k /
	\eta_p)$.
	\hfill{}
\end{proof}

To solve the binary instance of \tops problem, the FM sketches may be employed
while running the \incg algorithm on the cluster representatives. The resulting
scheme is referred to as \fmnc.  In that case, the bound is updated as follows.

\begin{thm}
	The approximation bound of \fmnc for the binary instance of \tops is $(k /
	\eta_p).(1 + \epsilon)^k$, where $\epsilon$ is the error parameter provided
	by the FM sketch.
\end{thm}

\begin{proof}
	If the error parameter for FM sketch is $\epsilon$, running it for $k$
	iterations produces an error bound of at most $(1 + \epsilon)^k$.  In
	conjunction with Th.~\ref{thm:topsbound}, the required error bound
	is obtained.
	\hfill{}
\end{proof}


\noindent
\textbf{Complexity Analysis:} 
For a given value of $\tau$, suppose the index instance $\inst_p$ with
$|\inst_p| = \eta_p$ clusters is used. Assume that the largest number of
trajectories passing through a cluster is $\xi_p = \max\{|\tl(g_i)|\}$, and
$\lambda_p$ is the largest number of vertices in any cluster in $\inst_p$.

\begin{thm}
	\label{thm:complexity}
	The time and space complexities of \nc are $O(k.\eta_p.\xi_p)$ and
	$O(\sum_{p=1}^t (\eta_p(\xi_p + \lambda_p)))$ respectively.
\end{thm}

The proof is provided in Appendix~\ref{app:netclus-complexity}.

The average values of $\eta_p,\lambda_p$ and $|\tl|$ are shown in
Table~\ref{tab:cluster},  for different values
of cluster radius $R_p$.

\section{Handling Dynamic Updates}
\label{sec:updates}

In this section, we discuss how the \nc framework efficiently handles dynamic
updates of trajectories and candidate sites. We assume that the underlying road
network does not change.  In each of the following cases, the updates are
processed for \emph{all} the index instances (of varying cluster radii).

\noindent
\textbf{Addition of a site:} Suppose a location $s_{add}$ is identified as a new
candidate site, i.e., $s_{add}$ gets added to $\mathcal{S}$. If $s_{add}$ is
already in $V$, its cluster $g_{add}$ is identified.  Otherwise, it is added to
the cluster $g_{add}$ whose cluster center $c_{add}$ is the closest.  To
determine the closest cluster center, the neighbors $N(s_{add})$ of $s_{add}$ in
$G$ are used.  The round-trip distance to a cluster center $c_i$ is estimated
using $\min_{s_l \in N(s_{add})} \{ \dr(s_{add}, s_l) + \dr(s_l, c_i)\}$ if
$\dr(s_l, c_i)$ is available.  Suppose $c_{near}$ is the nearest cluster
center to $s_{add}$. If the distance $\dr(s_{ad},g_{nearest})>2R_p$, then we
create a new cluster $g_{add}$ with $s_{add}$ as its center.  If the identified
cluster $g_{add}$ does not have a cluster-representative, then $s_{add}$ is
marked as its new cluster representative.  Else, it is determined if $s_{add}$
can be a better representative for $g_{add}$ as discussed in
Sec.~\ref{sec:representative}.  Finally, the exact round-trip distance to the
cluster center, $\dr(s_{add}, c_{add})$, is computed.

\noindent
\textbf{Deletion of a site:} Suppose a particular site $s_{del}$ is no longer
viable for a service and, therefore, needs to be deleted from $S$.  Suppose,
$s_{del}$ lies in the cluster $g_{del}$.  First, it is untagged as a candidate
site in $g_{del}$.
If $s_{del}$ is not the cluster representative of $g_{del}$, nothing more needs
to be done.  Otherwise, another candidate site, if available, is chosen as the
new cluster representative using the methodology described in
Sec.~\ref{sec:representative}.

\noindent
\textbf{Addition of a trajectory:} Suppose a new trajectory $T_{add}$ is added.
It is first mapped into a sequence of clusters, $g_1, \dots, g_l$.  For each
such cluster $g_i$, $T_{add}$ is added to the set $\tl(g_i)$ and its round-trip
distance to the cluster center $c_i$ of $g_i$, $\dr(T,c_i)$, is computed and
stored.  In addition, $g_i$ is added to the set $\cc(T_{add})$.  The procedure
is essentially the same one discussed in Sec.~\ref{sec:online}.
 
\noindent
\textbf{Deletion of a trajectory:} Suppose a trajectory $T_{del}$ is deleted.
Assume that the coverage set of $T_{del}$ is $\cc(T_{del}) = \{g_1,\dots,g_l\}$.
For each such cluster $g_i$, $T_{del}$ is removed from its coverage set
$\tl(g_i)$.   Finally, the set $\cc(T_{del})$ is deleted.

While multiple updates can be applied one after another, batch processing is
more efficient.  Sec.~\ref{sec:updateexp} shows that the updates are handled
quite efficiently.

\section{Extensions and Variants of \tops}
\label{sec:variants}

In this section, we present a few extensions and variants within the \tops
framework and discuss how the \incg algorithm for \tops can be adapted to solve
these problems.  As \nc essentially runs \incg on the cluster representatives,
it can also be adapted in a similar manner.

\subsection{Cost Constrained \tops}
\label{sec:cost}

In this problem, referred to as \topscost, each site $s_i \in \mathcal{S}$ has
a cost $cost(s_i)$ associated with it, and the goal is to select a set of sites
within a fixed budget $B$ such that the sum of trajectory utilities is
maximized. Formally, the problem is stated as follows. 

\begin{prob}[\topscost]
	Given a set of trajectories $\mathcal{T}$, a set of candidate sites
	$\mathcal{S}$ where each site $s_i \in \mathcal{S}$ has a fixed cost
	$cost(s_i)$, \topscost problem with query parameters $(B,\tau, \pref)$
	seeks to report a set $\mathcal{Q} \subseteq \mathcal{S}$, that maximize
	the sum of trajectory utilities, i.e., $\max_{\q \subseteq \s}
	U(\mathcal{Q}) = \sum_{T_j \in \traj} U_j$ such that $U_j = \max_{s_i \in
	\q}\{\pref(T_j,s_i)\}$,  $cost(\mathcal{Q})=\sum_{s_i \in \mathcal{Q}}
	cost(s_i) \le B$.
\end{prob}

The \tops problem reduces to \topscost by assigning unit cost to each site and
$B=k$.  In contrast to \tops, \topscost does not restrict the number of sites
selected in the answer set.  Since \tops is NP-hard, so is \topscost.

The \incg algorithm can be adapted based on the greedy heuristic for the
budgeted maximum coverage problem \cite{khuller1999budgeted} to solve
\topscost.  The algorithm starts with an empty set of sites $\mathcal{Q} =
\varnothing$ and proceeds in iterations. In each iteration, it selects a site
$s_i \in \mathcal{S} - \mathcal{Q}$ such that $\left(U(\mathcal{Q} \cup
\{s_i\}) - U(\mathcal{Q})\right)/cost(s_i)$ is \emph{maximal}. If $cost(s_i)$
is within the remaining budget $B - cost(\mathcal{Q})$, it is added to
$\mathcal{Q}$; otherwise, it is pruned from $\mathcal{S}$.  This process
continues until $\mathcal{S} = \varnothing$.

It was shown in \cite{khuller1999budgeted} that the above approach can perform
arbitrarily bad. Thus, in order to bound the approximation guarantee, the
algorithm is augmented with the following step. Assume $s_{max}$ to be a
candidate site such that $cost(s_{max}) \leq B$ and $U(\{s_{max}\})$ is
maximal.  The algorithm returns either the site $s_{max}$ or the set
$\mathcal{Q}$ whichever offers the maximum utility.  Following the analysis in
\cite{khuller1999budgeted}, this scheme is guaranteed to produce a solution
with an approximation bound of $(1-1/e)/2$.

\subsection{Capacity Constrained \tops}
\label{sec:capacity}

In this problem, referred to as \topscap, each site $s_i \in \mathcal{S}$ has a
fixed capacity $cap(s_i)$ that denotes the maximum number of trajectories it
can serve. The goal is to select a set $\q \subseteq \s$ of size $k$ such that
the sum of trajectory utilities is maximized.  Formally, the problem is stated
as follows. 

\begin{prob}[\topscap]
	Consider a set of trajectories $\mathcal{T}$, a set of candidate sites
	$\mathcal{S}$ where each site $s_i \in \mathcal{S}$ can serve at most
	$cap(s_i)$ trajectories. For any set $\q \subseteq \s$, let $x_{ji}$ be a
	boolean indicator variable such that $x_{ji}=1$ if and only if the
	trajectory $T_j$ can be served by the site $s_i \in \q$.  \topscap problem
	with query parameters $(k,\tau,\pref)$ seeks to report a set $\mathcal{Q}
	\subseteq \mathcal{S}, |\mathcal{Q}|=k$, that maximizes the sum of
	trajectory utilities, i.e., $\max_{\q \subseteq \s} U(\q) =  \sum_{T_j \in
	\traj} U_j$ such that $U_j = \max_{s_i \in \q} (\pref(T_j,s_i) . x_{ji})$,
	and $\forall s_i \in \q, \ \sum_{T_j \in \traj} x_{ji} \leq cap(s_i)$. 
\end{prob}

\tops reduces to \topscap by assigning the capacity of each site to be infinite
or more than the total number of trajectories in the dataset. Hence, \topscap
is also NP-hard.

\incg can be adapted to solve \topscap in the following manner. The algorithm
starts with an empty set $\q_0=\varnothing$ and sets $\forall T_j \in \traj,
U_j=0$. Suppose the set of sites selected after iteration $\theta=1,\dots,k$, is
denoted by $\q_\theta$.  In each iteration $\theta$, it augments $\q_{\theta-1}$
by selecting a site $s_\theta \in \s - \q_{\theta-1}$ that offers the maximal
marginal gain in utility.

It then updates the trajectory utilities $U_j$.  The utility of $T_j$ due to
$\q_{\theta}$ is $U_j(\q_\theta) = \max_{s_i \in \q_{\theta}}\pref(T_j,s_i)$.
The marginal gain in the utility of $T_j$ due to addition of $s_\theta$ is
$U_j(s_\theta) = U_j(\q_{\theta-1}\cup \{s_\theta\})-U_j(\q_{\theta-1})$.
Since any site $s_i \in \s -\q_{\theta-1}$ can serve at most
$\alpha_i=\min\{|TC(s_i)|, cap(s_i)\}$ trajectories, its marginal utility is
defined as the sum of the largest $\alpha_i$ trajectory marginal utilities
$U_j(s_i)$.

Since the objective function of the \topscap problem is identical to that of
\tops, it follows the non-decreasing sub-modular property.  Thus, \incg offers
the same approximation bound of $1-1/e$, as stated in
Lem.~\ref{lem:inc_topso2}.

\subsection{\tops with Existing Services}
\label{sec:existing}

Optimal location queries usually factor in existing service locations before
identifying new ones.  The problem is NP-hard.  The \incg algorithm can take
into account the existing service locations as follows.

Suppose $\mathcal{E_S}$ is the set of existing service locations.  On receiving
the query parameter $\tau$, the covering sets \tc and \stc over the set of
sites  $\mathcal{S} \cup \mathcal{E_S}$ are computed. \incg starts with
$\mathcal{Q}_0 = \mathcal{E_S}$ and updates the marginal utilities of the sites
in $\mathcal{S}$.  The remaining algorithm stays unaltered.  The algorithm
terminates after selecting $k$ sites from the set $\mathcal{S}$, in the same
manner as \tops. 

An advantage and important feature of \incg is that the site chosen in a given
iteration depends solely on what the existing service locations are, and not on
how they were chosen.

Since the initial utility  $U(\q_0)=U(\mathcal{E_S}) \neq 0$, the approximation
bound of $1-1/e$ is not directly applicable. However, we next show that the
same bound holds.

For any set $\q \subseteq \s$, the extra utility, defined as $U^\prime(\q) =
U(\q)-U(\mathcal{E_S})$ where $U(\mathcal{E_S})$ is the utility offered by the
existing services, is non-negative.  Now we have $U^\prime(\q)$ as a
non-decreasing sub-modular function with $U^\prime(\q_0) = 0$. Let $\opt$ and
\q denote the set of sites returned by an optimal algorithm and \incg
respectively. Then, following Lem.~\ref{lem:inc_topso2}, $U^\prime(\q) \geq
(1-1/e) U^\prime(\opt)$.  This leads to $U(\q) \geq (1-1/e) U(\opt) +
U(\mathcal{E_S})/e$. Since $U(\mathcal{E_S}) \geq 0$, hence, $U(\q) \geq
(1-1/e) U(\opt)$.

\subsection{Other \tops Variants}
\label{sec:other}

There are certain variants of trajectory-aware service location problems that
already exist in the literature. The proposed \tops formulation generalizes
these variants and, thus, the \incg algorithm can be used to solve them.  We
next discuss some of the important variations.

\noindent
\textbf{\topsone: Binary world}: This is the simple binary instance of
\tops query that is defined in Def.~\ref{def:binary}. 
The problem is still NP-hard \cite{berman1995locatingMain}. \incg offers the
same approximation bound for \topsone, as in the case of \tops
(Th.~\ref{thm:inc_topso}).  

\noindent
\textbf{\topstwo: Maximize market size:} Instead of operating in a
binary world where a trajectory is either covered or not covered, in this
formulation, the aim is to maximize the probability of a trajectory being
covered \cite{berman1995locatingMain}. The probability $p(T_j,s_i)$ of $T_j$
being covered by site $s_i$ is modeled as a convex function of the distance
$\dr(T_j,s_i)$ between them.  This problem aims to set up $k$ services that
maximizes the expected number of total trajectories served.  This is a special
case of our proposed \tops formulation with $\pref(T_j,s_i) = p(T_j,s_i)$ if
$\dr(T_j,s_i) \leq \tau$ and $0$ otherwise.  It is again NP-hard
\cite{berman1995locatingMain}. \incg offers the same approximation bound for
\topstwo, as in the case of \tops (Th.~\ref{thm:inc_topso}).

\noindent
\textbf{\topsthree: Minimize user inconvenience:} Assuming that each user on
its trajectory would necessarily avail a service, this problem aims to minimize
the expected deviation incurred by a user \cite{berman1995locatingMain}. This
can be handled through the \tops framework by setting the preference score to
$\pref(T_j,s_i) = - \dr(T_j,s_i)$, and $\tau$ to $\infty$.  Since the
trajectory utility is defined to be the maximum of the preference scores of the
selected sites, maximizing the sum of trajectory utilities would minimize the
total user deviation. \topsthree is NP-hard owing to reduction from the
$k$-medians problem \cite{berman1995locatingMain}. The approximation bound of
\incg for this problem is not yet known.


\noindent
\textbf{\topsfour: Best service locations under fixed market share:}
The aim here is to place the minimum number of services that can capture a fixed
share of the market comprising of the user trajectories
\cite{berman1992optimal}. The problem is complementary to \topsone and asks the
following: What is the smallest set of sites $\mathcal{Q} \subseteq \mathcal{S}$
such that at least $\beta$ fraction of $|\mathcal{T}|$ is covered, where $0 <
\beta \leq 1$?  This problem is NP-hard \cite{berman1992optimal}.  Since \incg
algorithm is iterative, it selects as many sites that are necessary to cover the
desired fraction of users.
Note that the set cover problem directly reduces to this problem, and so does
the greedy heuristic for the set cover problem to \incg. Hence, \incg algorithm
offers the same approximation bound of $1 + \ln n$ for \topsfour.

\subsection{Generic Framework}

From the above discussions, it is clear that the \tops framework is highly
generic and can absorb many extensions and variations with little or no
modification in the prposed algorithms \incg and \nc.  Also, importantly, the
framework also enables \emph{combining} multiple extensions and variants. For
example, \topscost and \topscap extensions can be merged to create a new version
of \tops.  In Sec.~\ref{sec:results-tops-variants}, we discuss experimental
evaluation of some of the above extensions and variants.

\section{Experimental Evaluation}
\label{sec:exp}

\begin{table}[t]\small
\centering
\begin{tabular}{ccrr}
\hline
\bf Dataset & \bf Type & \bf \#Trajectories & \bf \#Sites \\
\hline
\bs & Real & 1,000 & 50 \\
\bl & Real & 123,179 & 269,686 \\
\hline
Bangalore & Synthetic & 9,950 & 61,563 \\
New York & Synthetic & 9,950 & 355,930 \\
Atlanta & Synthetic & 9,950 & 389,680 \\
\hline
\end{tabular}
\tabcaption{Summary of datasets.}
\vspace{-0.05in}
\label{tab:datasets}
\moveup
\end{table}

In this section, we perform extensive experiments to establish
(1)~\emph{Efficiency:} that \nc is efficient, practical and scales well to
real-life datasets, and (2)~\emph{Quality:} that \nc produces solutions that are
close to those of \incg, which serves as the \emph{baseline} technique.

Most of the results are shown for the binary version of the problem since it is
easier to comprehend.  Sec.~\ref{sec:results-tops-variants} shows the results
for various extensions and variants of \tops.

The experiments were conducted using Java (version 1.7.0) platform on an
Intel(R) Core i7-4770 CPU @3.40GHz machine with 32 GB RAM running Ubuntu 14.04.2
LTS OS.

\subsection{Evaluation Methodology} 
\label{sec:methodology}

\noindent
\textbf{Algorithms:} We evaluate the performance of three different algorithms to
address \tops: \ipopt, \incg, and \nc, which are henceforth referred to
as \opt,  \inc, and \nc respectively in text and in the figures.  The
variants of \inc and \nc, based on the FM sketches, are henceforth referred to
as \incfm and \ncfm respectively.  

To the best of our knowledge, there is no existing algorithm for \tops that
works on real trajectories on city-scale road networks.  The state-of-the-art
algorithm for \topsone problem, proposed in \cite{berman1995locatingMain}, when
adapted for our case, can be reduced to \incg. Hence, \incg acts as the baseline
algorithm for evaluation on \topsone.  

\noindent
\textbf{Variants:} The specific variant of \tops on which we evaluated most of
the experiments was the binary instance defined in Def.~\ref{def:binary} or
\topsone.  We choose to evaluate this particular variant of \tops due to three
reasons: (1)~Usually, such binary optimization problems are the worst case
instances of the general integer optimization problems and are, therefore,
hardest to approximate.  (2)~The site preference function is simple. (3)~It has
several applications in transportation science and operations research.
However, we also evaluate a few other \tops variants and extensions in
Sec.~\ref{sec:results-tops-variants}.
   
\noindent
\textbf{Metrics of evaluation:} The main metrics of evaluation were (a)~total
\emph{utility} measured as a percentage of the total number of trajectories $m$,
and (b)~query \emph{running time}.  The two basic parameters studied were
(i)~number of service locations $k$, and (ii)~coverage threshold $\tau$,
whose default values were $5$ and $0.8$\,Km. respectively.

We conducted experiments on both real and synthetic datasets, whose details are
shown in Table~\ref{tab:datasets}.  For simplicity, we assume that the number of
candidate sites is the same as the number of nodes in the graph, unless
otherwise stated.

\noindent
\textbf{Real datasets:} We used GPS traces of taxis from Beijing consisting
of user trajectories generated by tracking taxis for a week \cite{cab1,cab2}.
This is the most widely used and one of the largest publicly available
trajectory datasets.  To generate trajectories as sequences of road
intersections, the raw GPS-traces were map-matched~\cite{map1} to the Beijing
road network extracted from OpenStreetMap (\url{http://www.openstreetmap.org/}).
The road network contains 269,686 nodes and 293,142 edges, with an underlying
area of $\sim$1720 sq.\,Km.


Since \tops is NP-hard, the optimal algorithm requires exponential time and,
therefore, can be run only on a very small dataset. Hence, we evaluate all the
algorithms against the optimal on \bs which is generated by randomly sampling
1000 trajectories from a fixed area, and then randomly selecting the set
$\mathcal{S}$, consisting of 50 candidate sites, from the same area.  The
sampling was conducted 10 times to increase the robustness of the results.  All
the other experiments were done on the full \bl dataset. 
 
\noindent
\textbf{Synthetic datasets:} To study the impact of city geographies, we
generated three synthetic datasets that emulate trajectories in the patterns
followed in New York, Atlanta and Bangalore.  We used an online traffic generator
tool, MNTG (\url{http://mntg.cs.umn.edu/tg/index.php}) to generate the traffic
traces, that were later map-matched to generate the trajectories in the desired
format.

\subsection{Choice of Parameters}
\label{sec:parameters}

We first run experiments to determine the choice of two important parameters:
(a)~the resolution of the index instances, $\eps$, and (b)~the number of FM
bit vectors, $f$.

As discussed earlier in Sec.~\ref{sec:index}, the choice of \eps affects both
the storage and offline run-time costs as well as quality.  Table~\ref{tab:eps}
lists the values for the \bl dataset when \eps is changed.  The error is
measured as  relative loss in utility of \nc  with that of \inc.  When \eps is too
small, there is almost no compression of the trajectories.  As a result, the
index structure size is large.  On the other hand, with a very large \eps, the
error may be unacceptable.  We fix $\eps=0.75$ for our experiments since it
offers a nice balance of a medium sized index structure that can fit in most
modern systems with an error of within 5\%.

\begin{table}[t]
\scriptsize
	\centering
		\begin{tabular}{c|c|c|c}
			\hline
			\eps	& Time (s)	& Space (GB)	& Rel. Error \% w.r.t. \inc  \\
			\hline
			0.25	& 108427	& 14.095	& 3.54 \\
			0.50	& 3216		& 4.215		& 3.97 \\
			0.75	& 1652		& 2.374		& 4.53 \\
			1.00	& 520		& 1.053		& 5.21 \\
			\hline 
		\end{tabular}
		\tabcaption{Variation across resolution of index instances, \eps.}
		\label{tab:eps}
\end{table}

%
Table~\ref{tab:f} shows how the utility and running time varies when $f$ number
of FM sketches are used as compared to the original \nc.  The error is measured as a relative loss in utility of \ncfm with that of \nc.
The values of $\tau$
and $k$ were their default values ($0.8$\,Km. and $5$ respectively).  When
$f$ is very small, the error is too large.  As $f$ increases, expectedly the
error decreases while the speed-up decreases as well.  When $f$ is extremely
high, the number of operations may overshadow the gains and using FM sketches
may be actually slower.  We fixed $f = 30$ since it produced less than $5\%$
error with a speed-up factor of more than $5$.

\begin{table}[t]
\scriptsize
	\centering
	\begin{tabular}{r|rr|r|rr|r}
	\hline
	\multirow{2}{*}{$f$}	& \multicolumn{2}{c|}{Utility}	&
	\multirow{2}{*}{Rel. Error \%}	& \multicolumn{2}{c|}{Time (ms)}	& \multirow{2}{*}{Speed-up} \\
	\cline{2-3}
	\cline{5-6}
	& \nc & FM &	& \nc & FM & \\	
	\hline
	1	& 47.23	& 26.60	& 43.67	& 846.65	& 16.32		& 51.88 \\
	2	& 47.23	& 34.27	& 27.45	& 846.65	& 21.66		& 39.09 \\
	4	& 47.23	& 39.33	& 16.73	& 846.65	& 32.65		& 25.93 \\
	10	& 47.23	& 41.27	& 12.62	& 846.65	& 65.32		& 12.96 \\
	20	& 47.23	& 43.29	& 8.34	& 846.65	& 116.32	& 7.28  \\
	30	& 47.23	& 44.96	& 4.81	& 846.65	& 161.53	& 5.24  \\
	40	& 47.23	& 45.52	& 3.63	& 846.65	& 216.62	& 3.91  \\
	50	& 47.23	& 45.89	& 2.84	& 846.65	& 272.18	& 3.11  \\
	100	& 47.23	& 46.43	& 1.69	& 846.65	& 984.17	& 0.86  \\
	\hline 
	\end{tabular}
	\tabcaption{Variation across the number of FM sketches, $f$.}
	\label{tab:f}
	\moveup
\end{table}

\subsection{Comparison with Optimal}
\label{sec:comparison with optimal}
\begin{figure}[t]
	\centering
	\moveup
	\moveup
	\subfloat[Utility.]
	{
		\includegraphics[width=\subfigwidth]{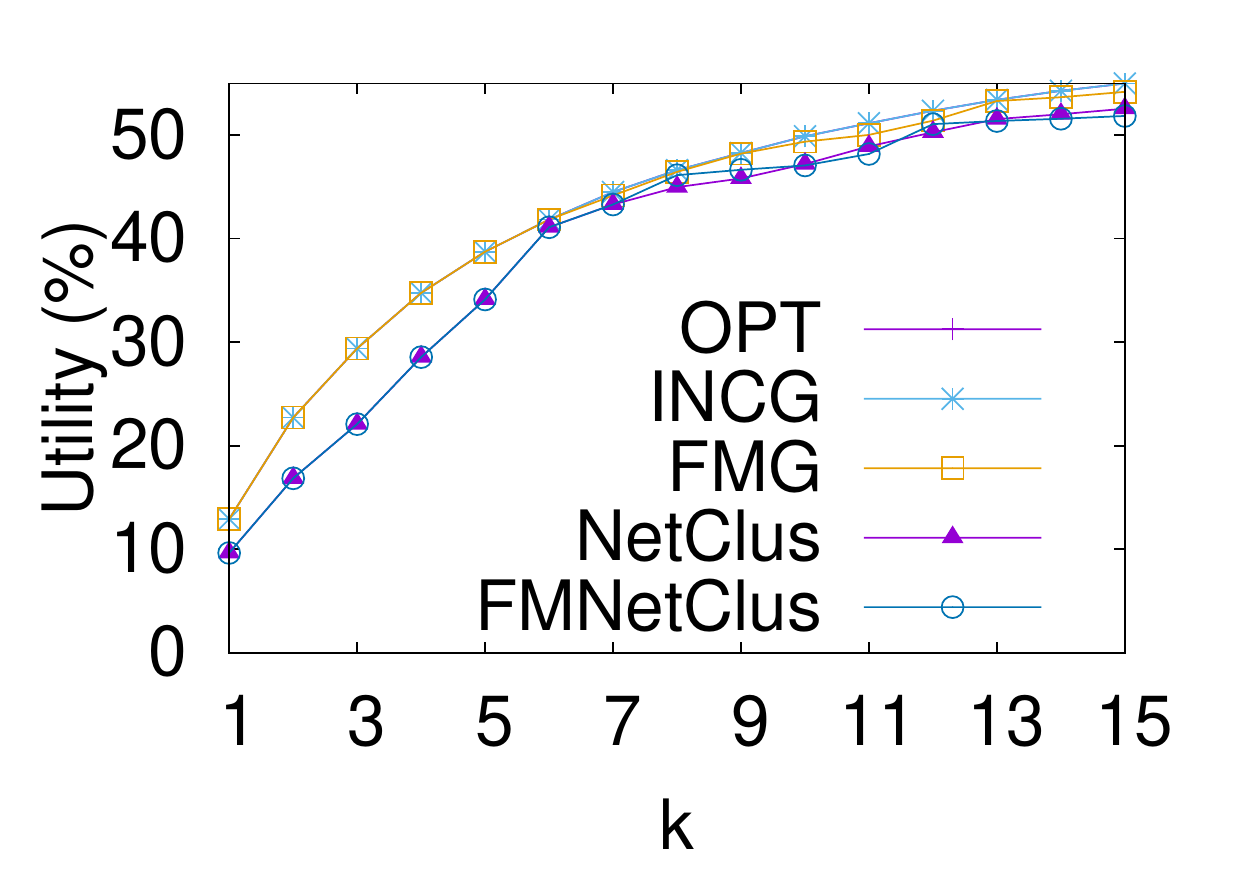}
		\label{subfig:optutil}
	}
	\subfloat[Running time.]
	{
		\includegraphics[width=\subfigwidth]{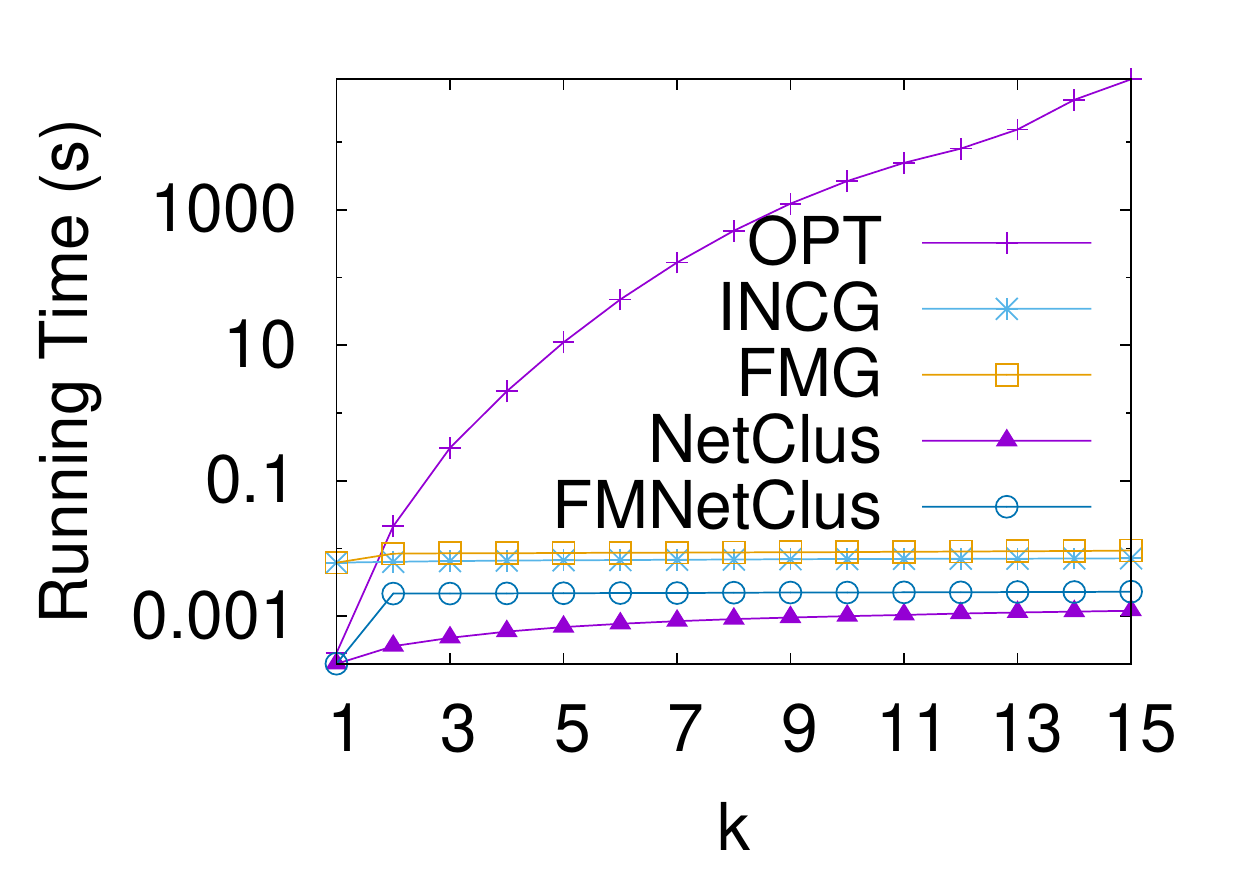}
		\label{subfig:opttime}
	}
	\figcaption{Comparison with optimal at $\tau=0.8$ Km.}
	\label{fig:optimal}
\end{figure}

Since this integer linear program based optimal algorithm requires impractical running
times, we ran it only on the \bs dataset mainly to assess the quality of the
other algorithms.  Fig.~\ref{fig:optimal} shows that the average utility of all
the algorithms are quite close to \opt although the running times are much
better.  (Note that the utilities in these and all subsequent figures are
plotted as a percentage of the total number of trajectories.)  \opt requires
hours to complete even for this small dataset and, therefore, is not practical
at all.  Consequently, we did not experiment with \opt any further.

\subsection{Quality Results}

\begin{figure}[t]
	\centering
	\moveup
	\moveup
	\subfloat[Varying $k$]
	{
		\includegraphics[width=\subfigwidth]{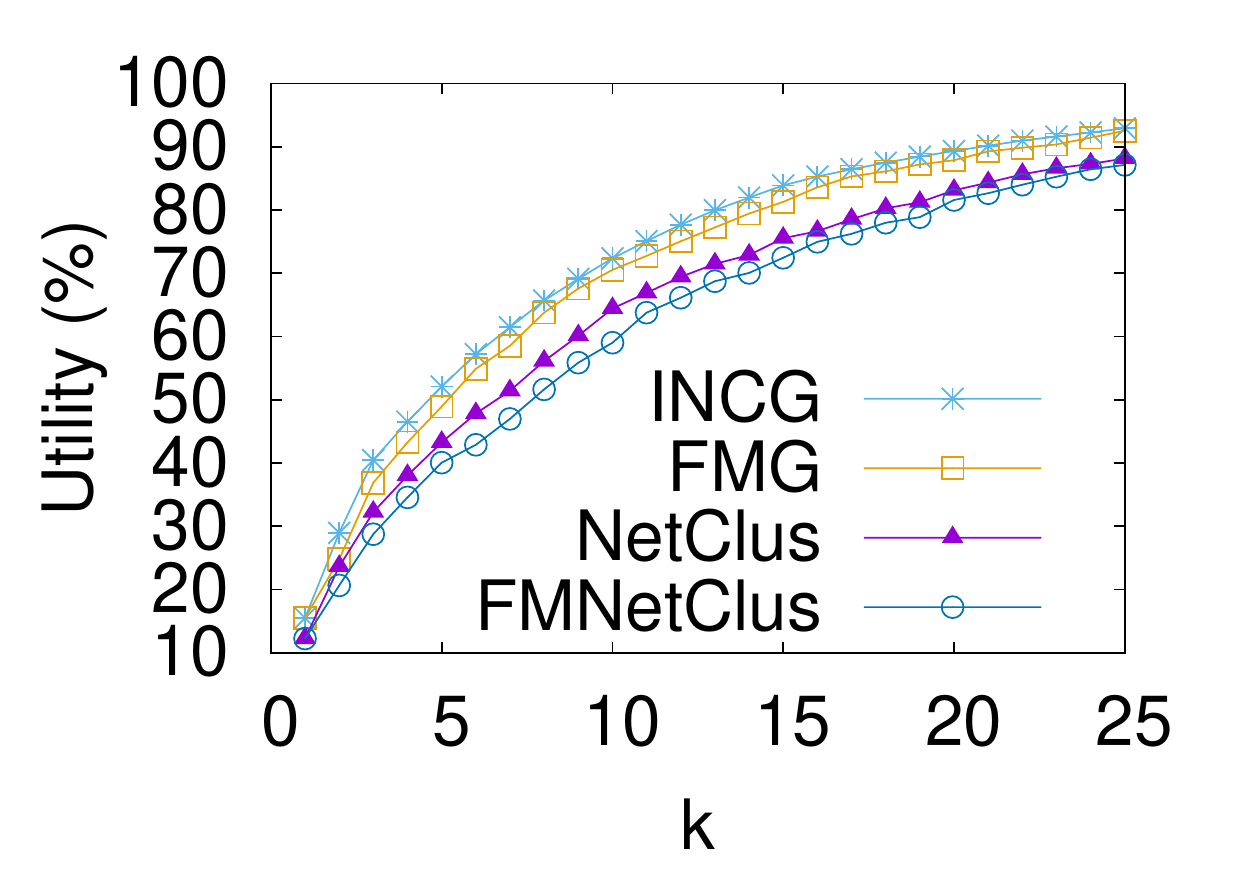}
		\label{subfig:simkutil}
	}
	\subfloat[Varying $\tau$]
	{
		\includegraphics[width=\subfigwidth]{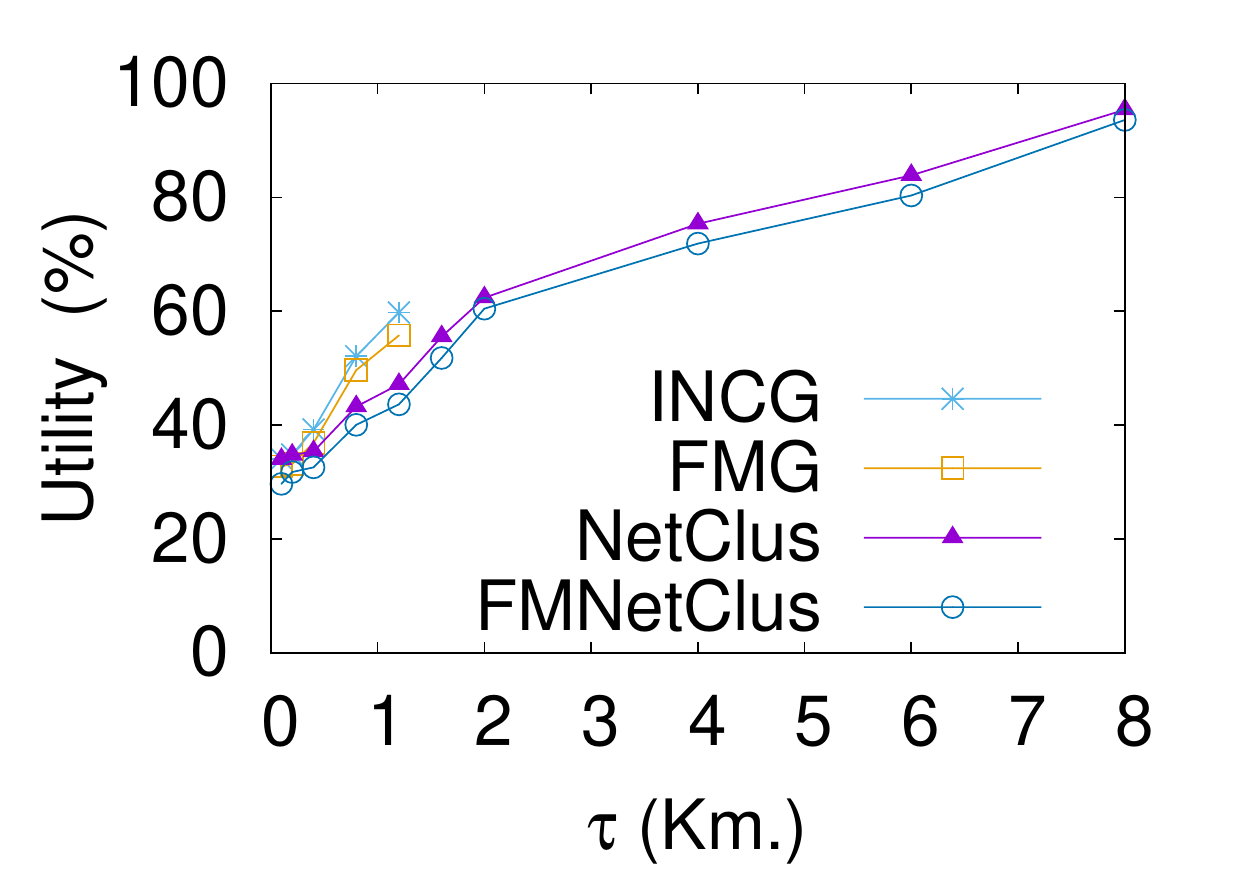}
		\label{subfig:simtauutil}
	}
	\figcaption{Quality results.}
	\label{fig:sim}
	\moveup
\end{figure}

Fig.~\ref{fig:sim} shows the
utility yields for different values of $k$ and $\tau$. The
utilities of \nc are close to that of \inc and are within 93\% of it on an average.  Owing to
high memory requirements (for reasons to be discussed in the next section),  \inc and \incfm could not run beyond $\tau=1.2$ Km.  The
utilities of \incfm and \ncfm are very close to that of \inc and \nc,
respectively.

\subsection{Memory Footprint}

\begin{table}[t]
\scriptsize
	\centering
	\begin{tabular}{|c|rr|cc|}
	\hline
	$\tau$ (in Km.)	& \inc	& \incfm & \nc & \ncfm \\
		\hline 
0.1	&	7.04	&	7.90	&	6.43	&	7.09	\\
0.2	&	9.14	&	10.00	&	4.17	&	4.81	\\
0.4	&	13.47	&	14.34	&	3.67	&	3.94	\\
0.8	&	19.58	&	20.44	&	3.22	&	3.64	\\
1.2	&	23.98	&	24.85	&	3.52	&	3.88	\\
1.6	&	\multicolumn{2}{c|}{Out of memory}	&	3.41	&	3.98	\\		
		\hline
	\end{tabular}
	\tabcaption{Memory footprint of different algorithms (in GB).}
	\label{tab:memory}
\end{table}

Table~\ref{tab:memory} shows that the memory footprints of \nc and \ncfm are
significantly less than those of \inc and \incfm.  As the coverage threshold
$\tau$ increases, the size of the covering sets, $\tc$ and $\stc$, used in \inc
and \incfm, increase sharply. Consequently, these algorithms could not scale
beyond $\tau=1.2$ Km. On the other hand, with higher $\tau$, \nc and \ncfm use
lower resolution clustering instances leading to higher data compression,
thereby resulting in lower memory footprints. The FM sketch based schemes
require slightly more memory than their counterparts, due to storage of multiple
bit vectors for each site or cluster, as applicable.

\subsection{Performance Results}

\begin{figure}[t]
	\centering
	\moveup
	\moveup
	\subfloat[Varying $k$]
	{
		\includegraphics[width=\subfigwidth]{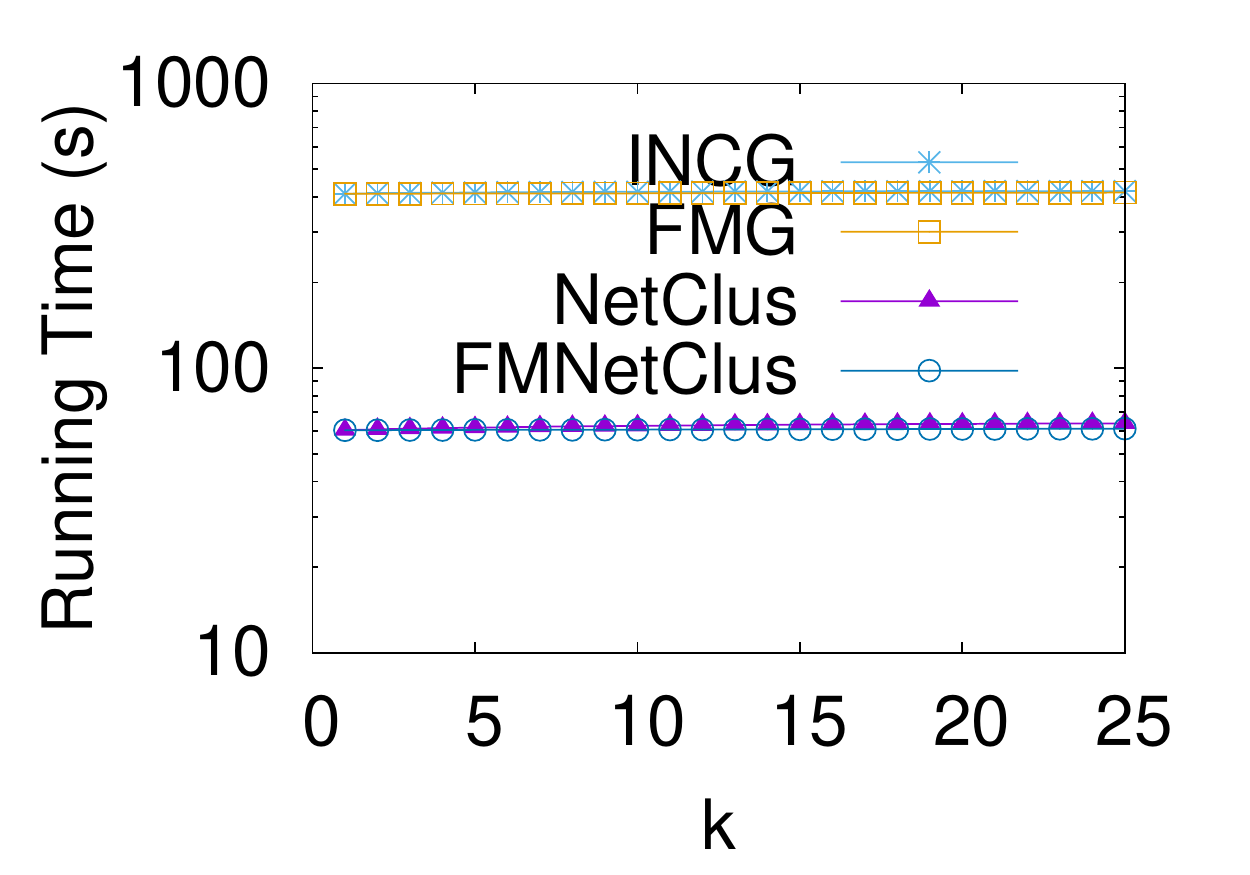}
		\label{subfig:ktime}
	}
	\subfloat[Varying $\tau$]
	{
		\includegraphics[width=\subfigwidth]{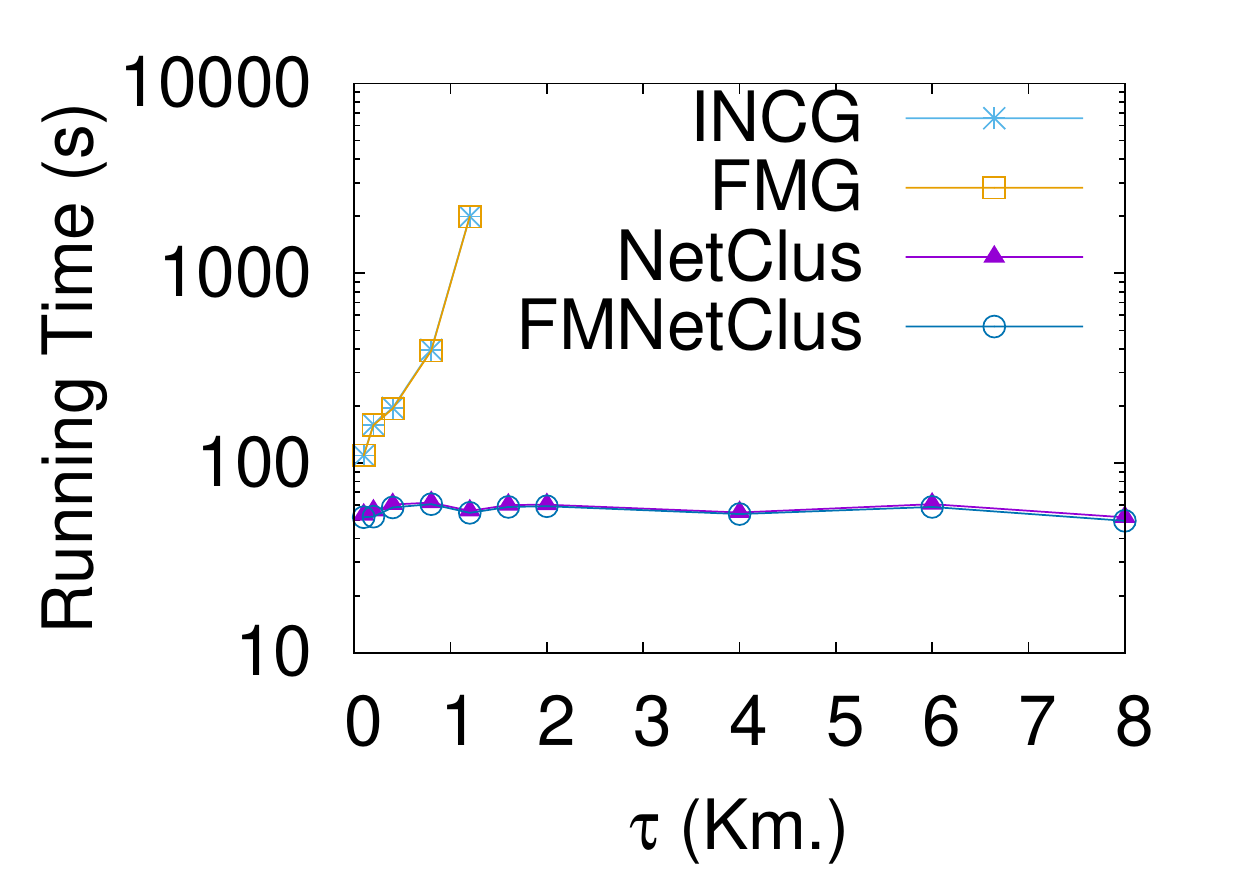}
		\label{subfig:tautime}
	}
	\figcaption{Running time performance.}
	\label{fig:time}
\vspace{-0.05in}
\moveups
\end{figure}

We next measure the performance of the algorithms for different values of $k$
and $\tau$.  Fig.~\ref{fig:time} shows that for $\tau \le 1.2$ km.,  \nc and
\ncfm are up to $36$ times faster than \inc and \incfm, respectively. For $\tau
>1.2$ km., as stated in the previous section, \inc and \incfm fail to run due to
high memory overheads.  When $\tau$ increases, \nc (and \ncfm) uses a higher
index instance having lesser number of clusters, leading to its efficiency. On
the other hand, \inc and \incfm use \tc and \stc covering sets of increasingly
larger size, resulting in poor performance.

Note that the plots in Fig.~\ref{subfig:ktime} appear to be linear w.r.t. $k$. This is because (i)~the initial cost of computing the covering sets significantly dominates the iterative phase of the algorithms, (ii)~the running times are plotted in log-scale. 
 
Although \ncfm (\incfm) offers a speed up of about $5$ times in the algorithm
running time when compared to \nc (\inc respectively), its effect is negated by
a relatively large initial pre-processing time required for computing the
covering sets. Due to this fact, we only compare the results of \nc with that of
\inc in the subsequent sections.

\subsection{Extensions and Variants of \tops}
\label{sec:results-tops-variants}

\begin{figure}[bt]
	\centering
	\moveup
	\subfloat[Cost.]
	{
		\includegraphics[width=\subfigwidth]{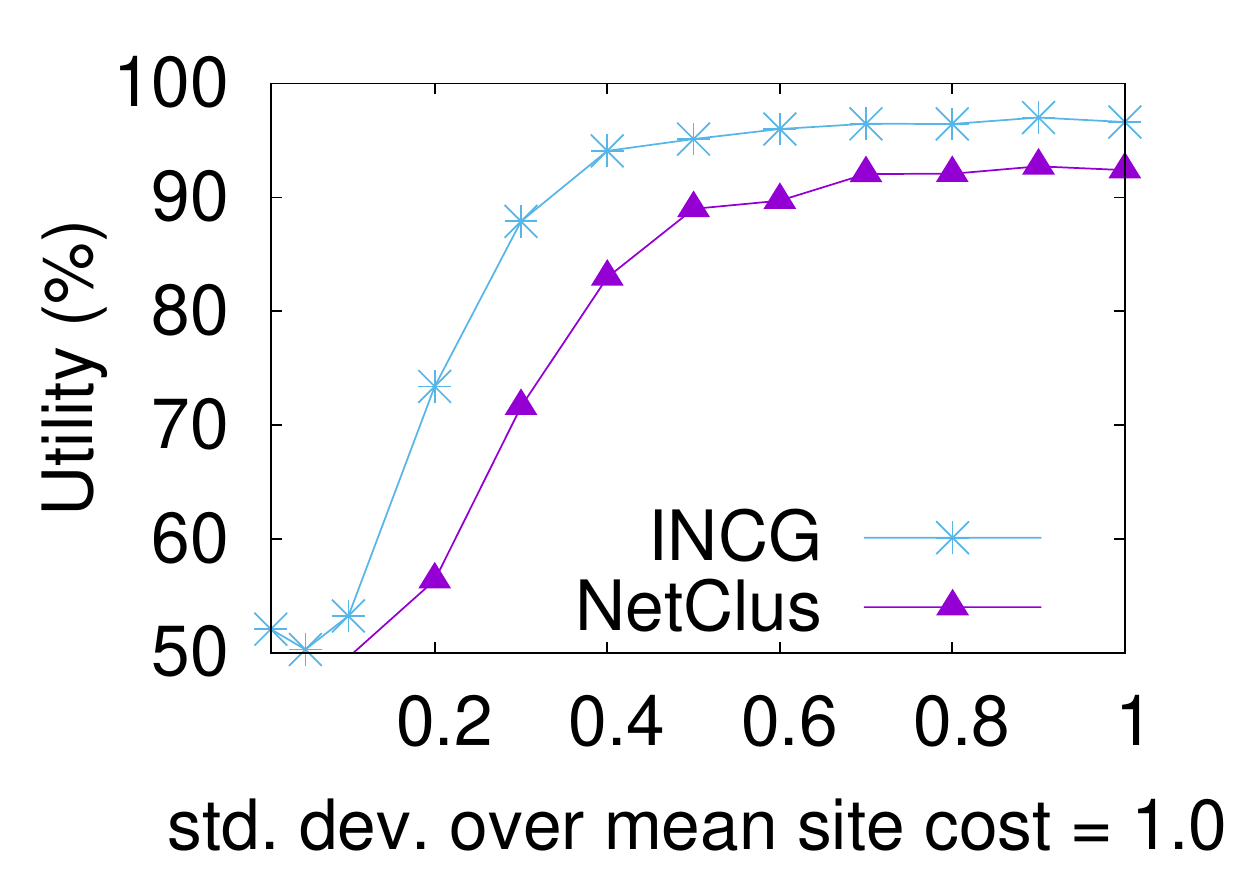}
		\label{subfig:cost_util}
	}
	\subfloat[Capacity.]
	{
		\includegraphics[width=\subfigwidth]{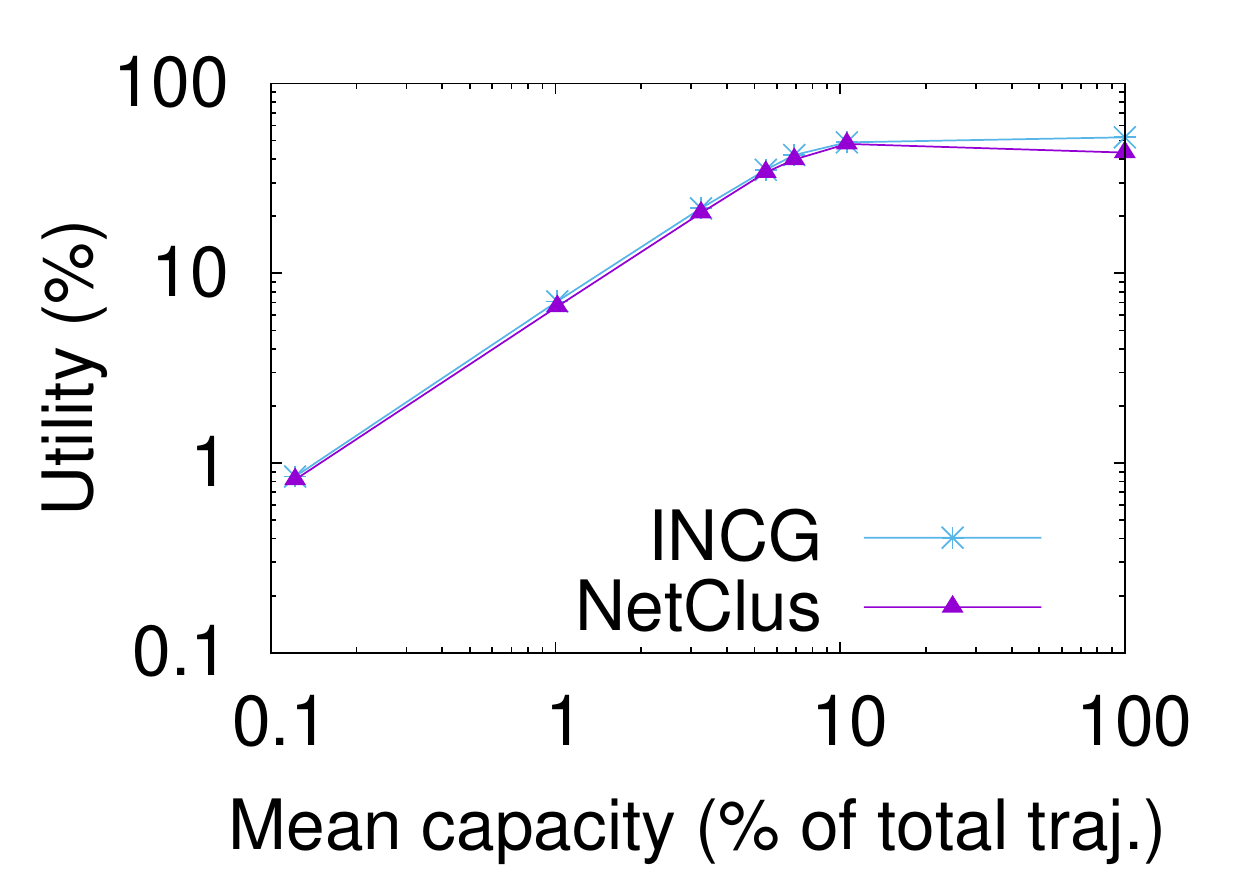}
		\label{subfig:cap_util}
	}
	\figcaption{Utilities for \tops extensions.}
	\label{fig:cost-cap-util}
\end{figure}

\begin{figure}
	\centering
	\moveups
	\subfloat[Utility.]
	{
		\includegraphics[width=\subfigwidth]{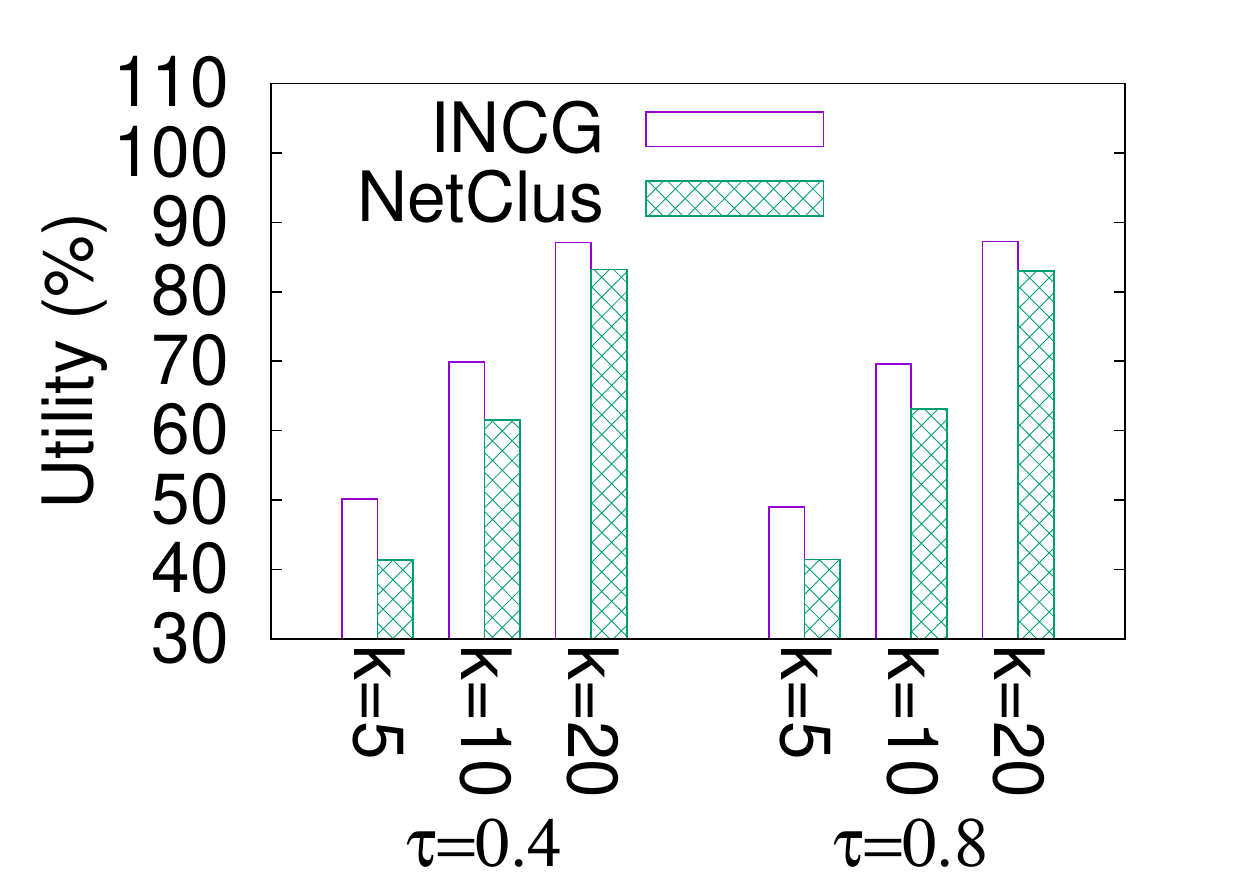}
		\label{subfig:variantutil}
	}
	\subfloat[Running time.]
	{
		\includegraphics[width=\subfigwidth]{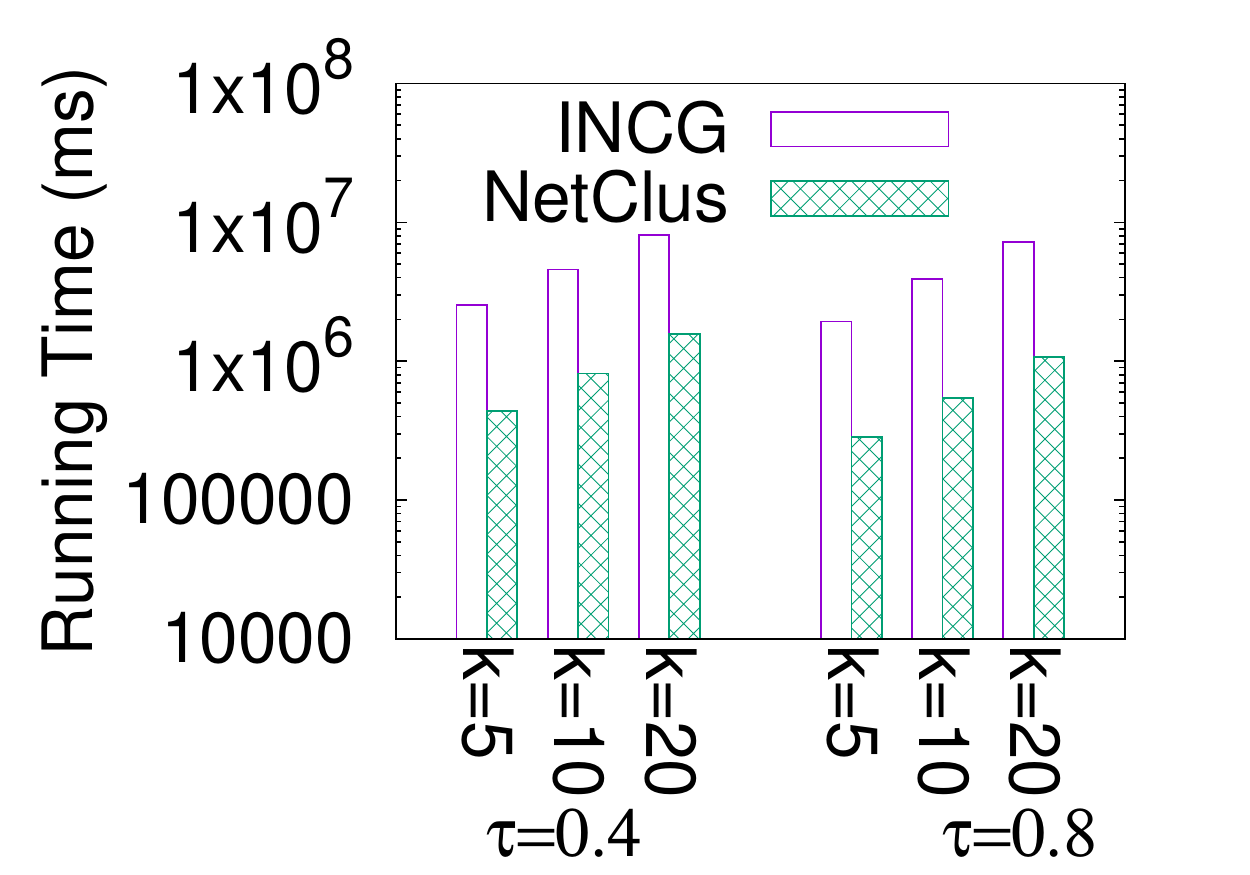}
		\label{subfig:varianttime}
	}
	\figcaption{\topstwo: A variant of \tops.}
	\label{fig:variant}
\end{figure}

We next show results of \nc on different \tops extensions and variants
(discussed in Sec.~\ref{sec:variants}) over the \bl dataset.

\textbf{\topscost:}
We consider a budget of $B=5.0$ and $\tau=0.8$ Km.  The cost of each site was
assigned using a normal distribution with mean $\mu = 1.0$ and standard
deviation varied between $\sigma \in [0,1]$ (the least cost of a site was
constrained to be $0.1$).  Fig.~\ref{subfig:cost_util} shows that the utility
increases with standard deviation (note that $\sigma = 0$ degenerates to basic
\tops).  This is due to the fact that with higher standard deviation, more
number of sites can be chosen with lower costs which ultimately leads to larger
number of trajectories being covered.  The increased number of iterations does
not increase the running time much since it is a small overhead on the initial
costs (Fig.~\ref{fig:cost2}).

\begin{figure}[t]
	\centering
	\moveups
	\subfloat[Number of sites.]
	{
		\includegraphics[width=\subfigwidth]{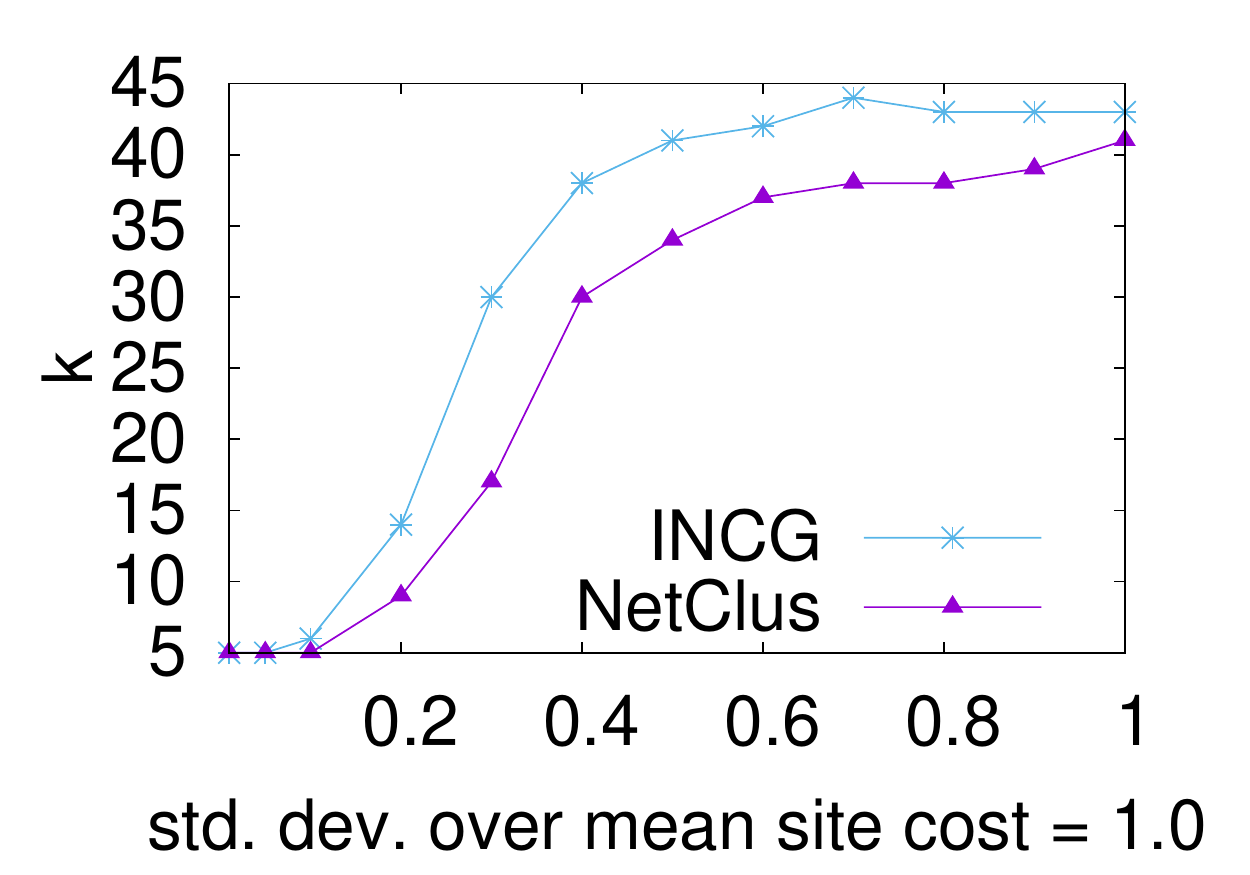}
		\label{subfig:cost_k}
	}
	\subfloat[Running time.]
	{
		\includegraphics[width=\subfigwidth]{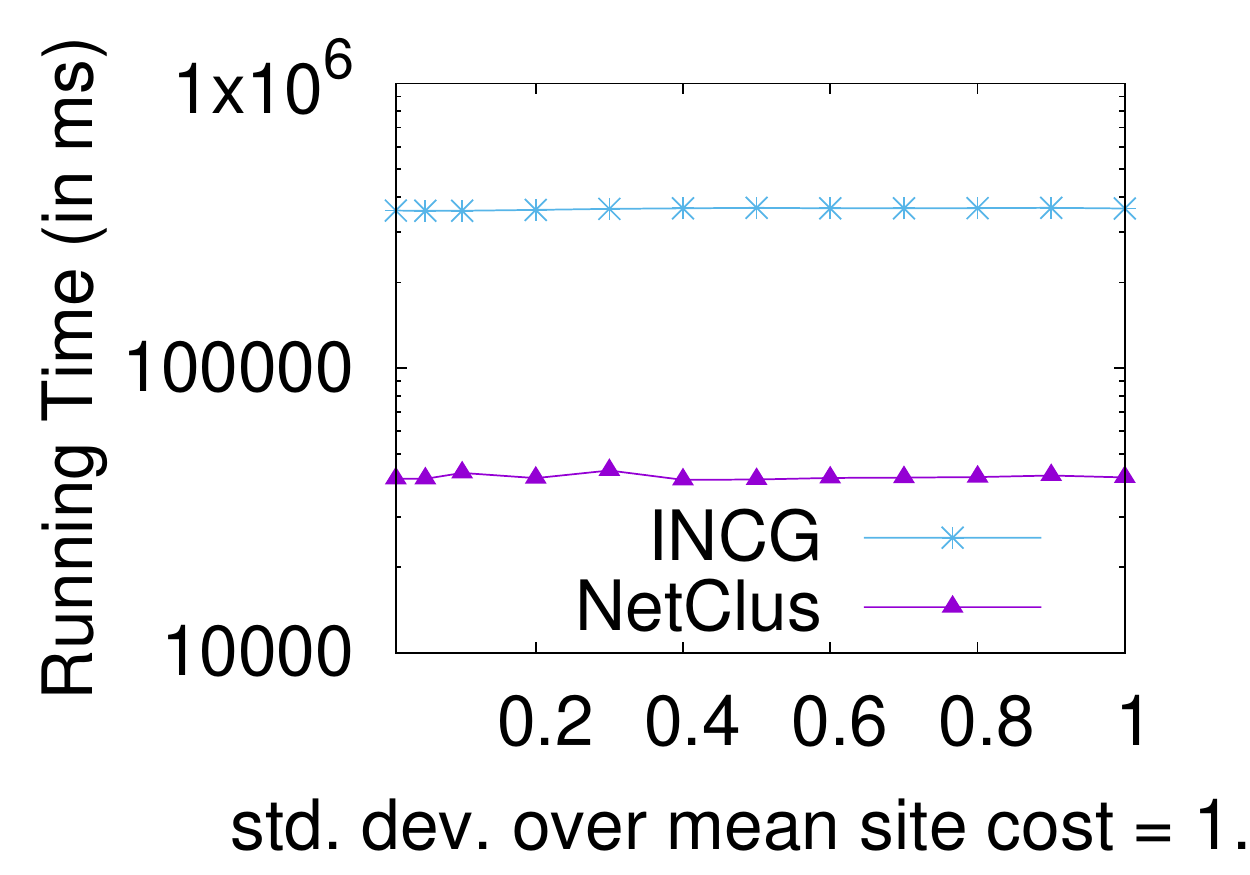}
		\label{subfig:cost_time}
	}
	\figcaption{\tops under cost constraint.}
	\vspace*{1mm}
	\label{fig:cost2}
\end{figure}

\textbf{\topscap:}
We consider $k=5$ and $\tau=0.8$ Km. The sites were assigned varying capacities
drawn from a normal distribution where the mean was varied in the range
$[0.1\%,100\%]$ of the total number of trajectories, and the standard deviation
fixed at $10\%$ of the mean. (note that mean capacity of $100\%$ corresponds to
basic unconstrained \tops).  Fig.~\ref{subfig:cap_util} shows that, as expected,
utility increases with mean capacity.  \nc has almost the same utility as that
of \inc. We do not show the running time plots, as the algorithms for \topscap
are almost the same as those for \tops and, hence, exhibit similar performance.

\textbf{\topstwo:}
Finally, we study the \topstwo variant where the preference function \pref was a
convex function of the distance between the site and trajectory.  The results,
shown in Fig.~\ref{fig:variant}, portray that \nc has utility close to that of
\inc, while being about an order of magnitude faster.

\subsection{Updates of Sites and Trajectories}
\label{sec:updateexp}

\begin{table}[t]
\scriptsize
\centering
\begin{tabular}{|c|r||c|r|}
\hline
\# Trajectories added & Update time & \# Candidate sites added & Update time \\
\hline
10000	&	 22.83 s	&	10000	&	1.26 s	\\
20000	&	 44.58 s	&	20000	&	1.56 s	\\
30000	&	 74.07 s	&	30000	&	1.76 s	\\
40000	&	 92.06 s	&	40000	&	1.85 s	\\
50000	&	122.69 s	&	50000	&	2.10 s	\\
\hline
\end{tabular}
\tabcaption{Index update cost.}
\label{tab:update}
\moveupb
\end{table}

Table~\ref{tab:update} shows that \nc efficiently processes additions of trajectories and candidate sites over the index structure 
(Sec.~\ref{sec:updates}). Adding a trajectory requires more
time than that for a candidate site since a trajectory passes through multiple
clusters in general and the covering sets, etc. of all those clusters need to be
updated.  Adding a site, on the other hand, requires simply finding the cluster
it is in and updating the cluster representative, if applicable.

\subsection{Robustness with Parameters}


\begin{figure}[t]
	\centering
	\moveups
	\subfloat[Number of sites.]
	{
		\includegraphics[width=\subfigwidth]{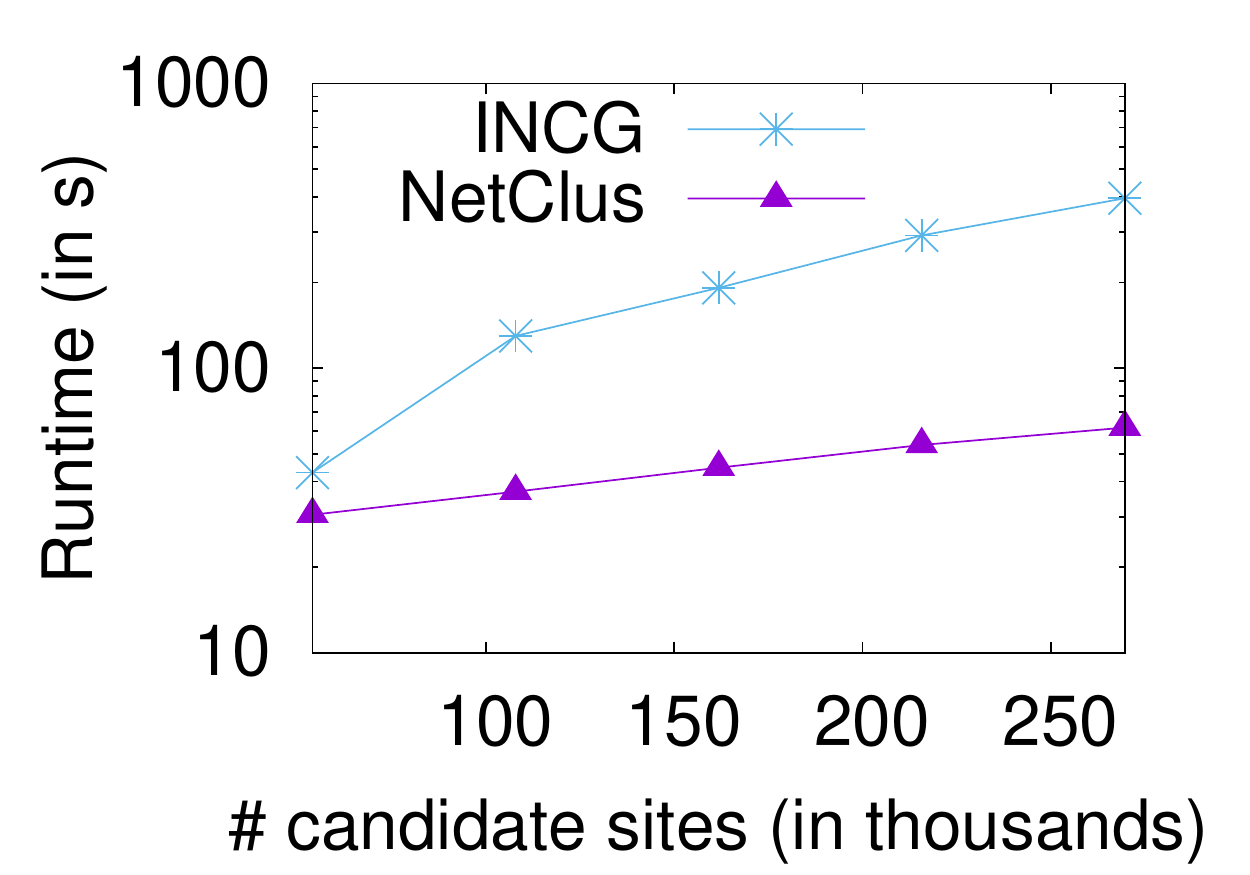}
		\label{subfig:site}
	}
	\subfloat[Number of trajectories.]
	{
		\includegraphics[width=\subfigwidth]{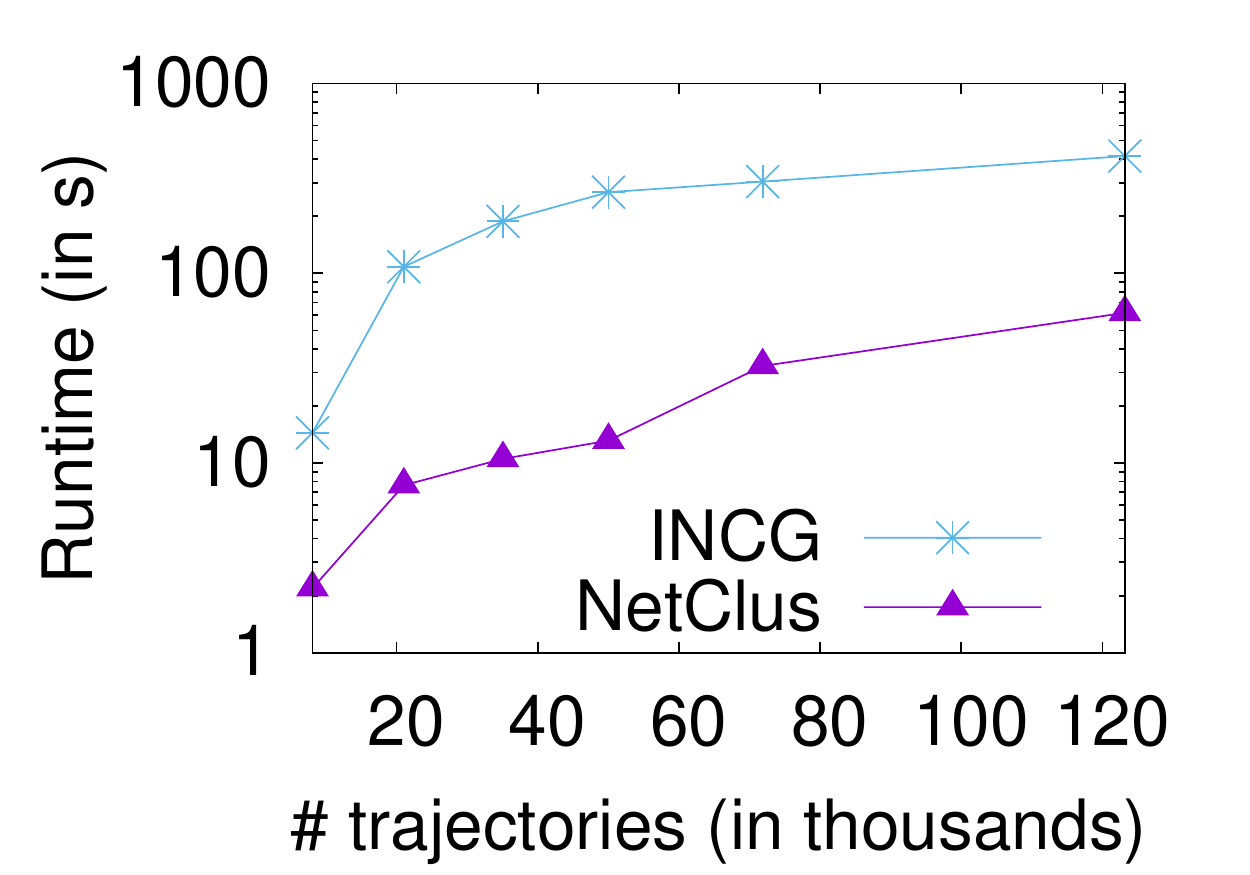}
		\label{subfig:traj}
	}
	\figcaption{Scalability results ($k=5$ and $\tau=0.8$ Km).}
	\label{fig:scale}
\end{figure}

\begin{figure}[tb]
	\centering
	\moveups
	\subfloat[Utility.]
	{
		\includegraphics[width=\subfigwidth]{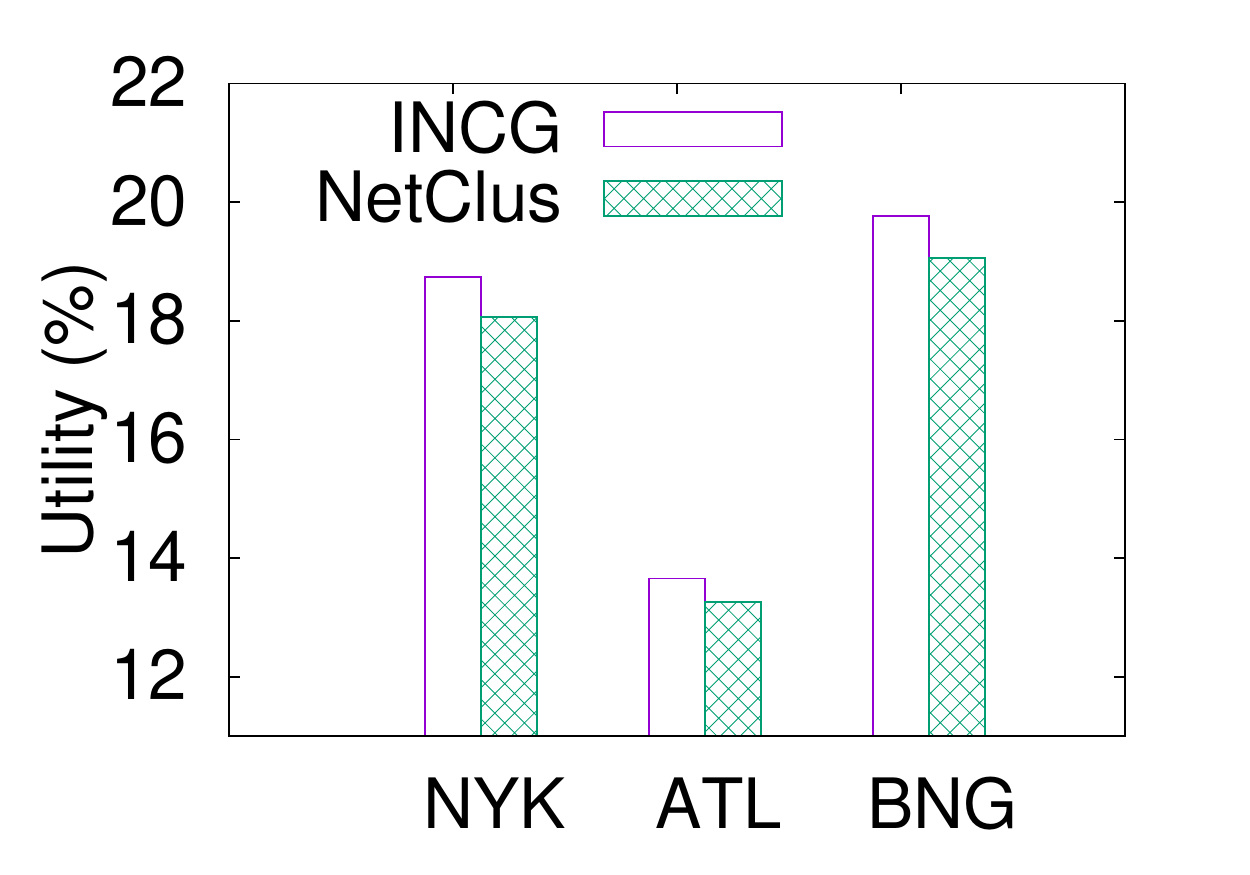}
		\label{subfig:synutil}
	}
	\subfloat[Running time.]
	{
		\includegraphics[width=\subfigwidth]{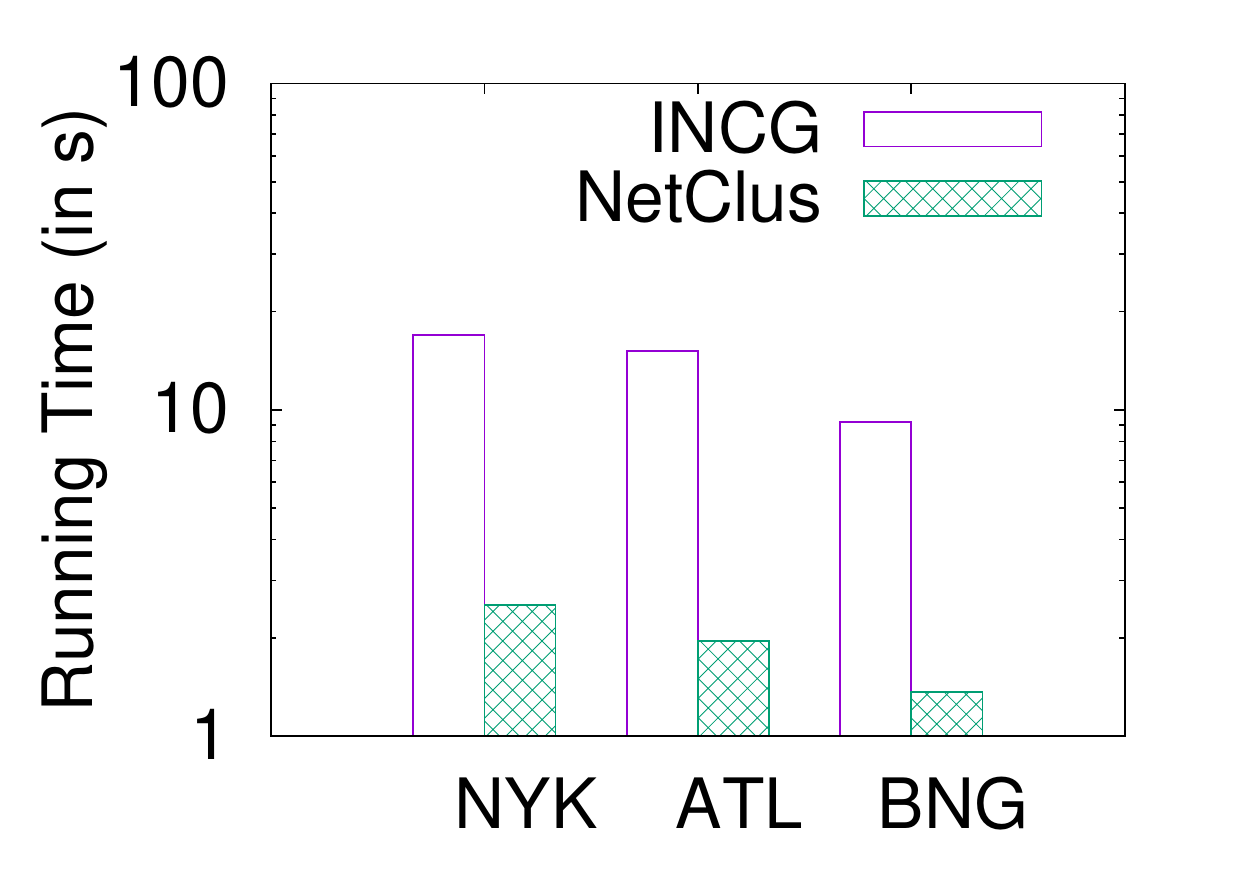}
		\label{subfig:syntime}
	}
	\figcaption{Effect of city geometries ($k=5$ and $\tau=0.8$ Km).}
	\label{fig:synthetic}
\end{figure}

\begin{figure}[tb]
	\centering
	\moveups
	\subfloat[Utility.]
	{
		\includegraphics[width=\subfigwidth]{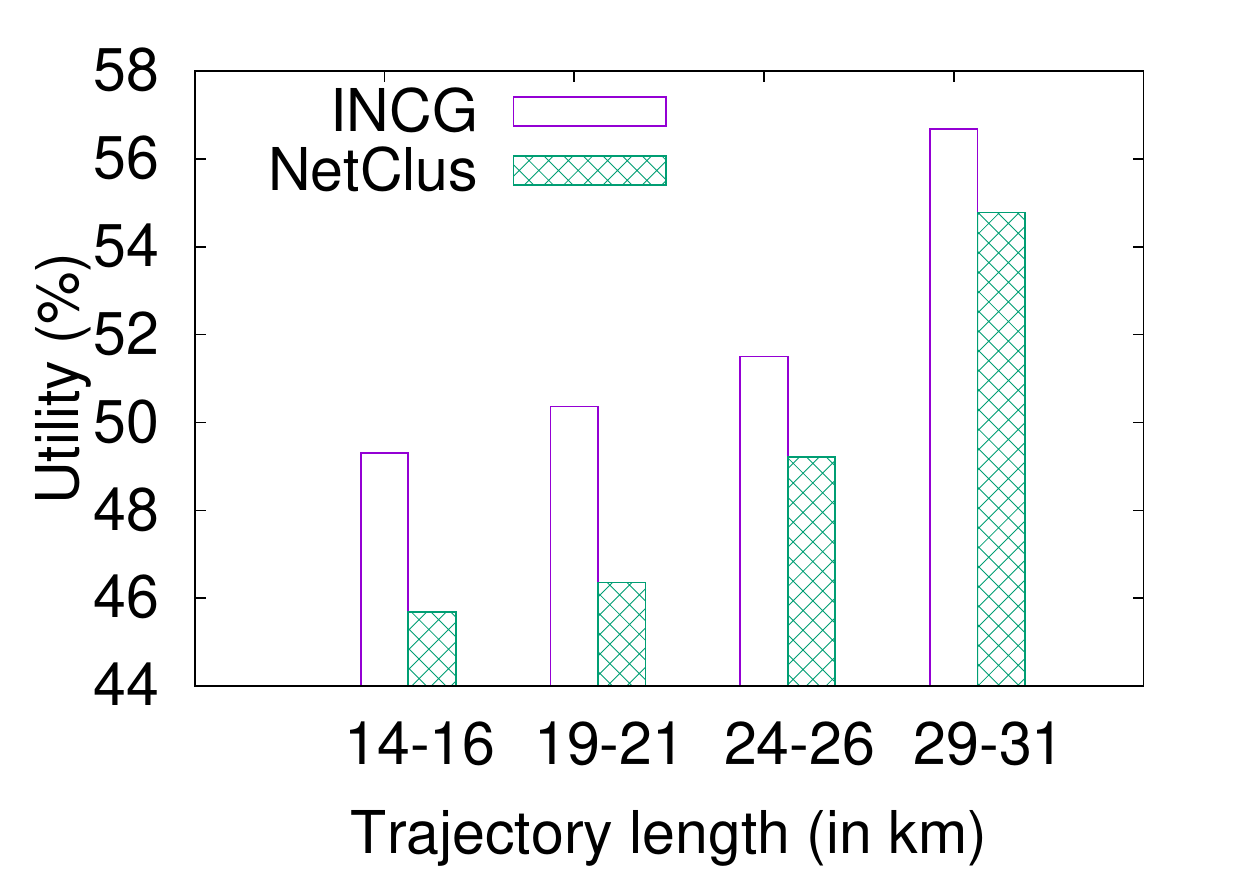}
		\label{subfig:lengthutil}
	}	
	\subfloat[Running time.]
	{
		\includegraphics[width=\subfigwidth]{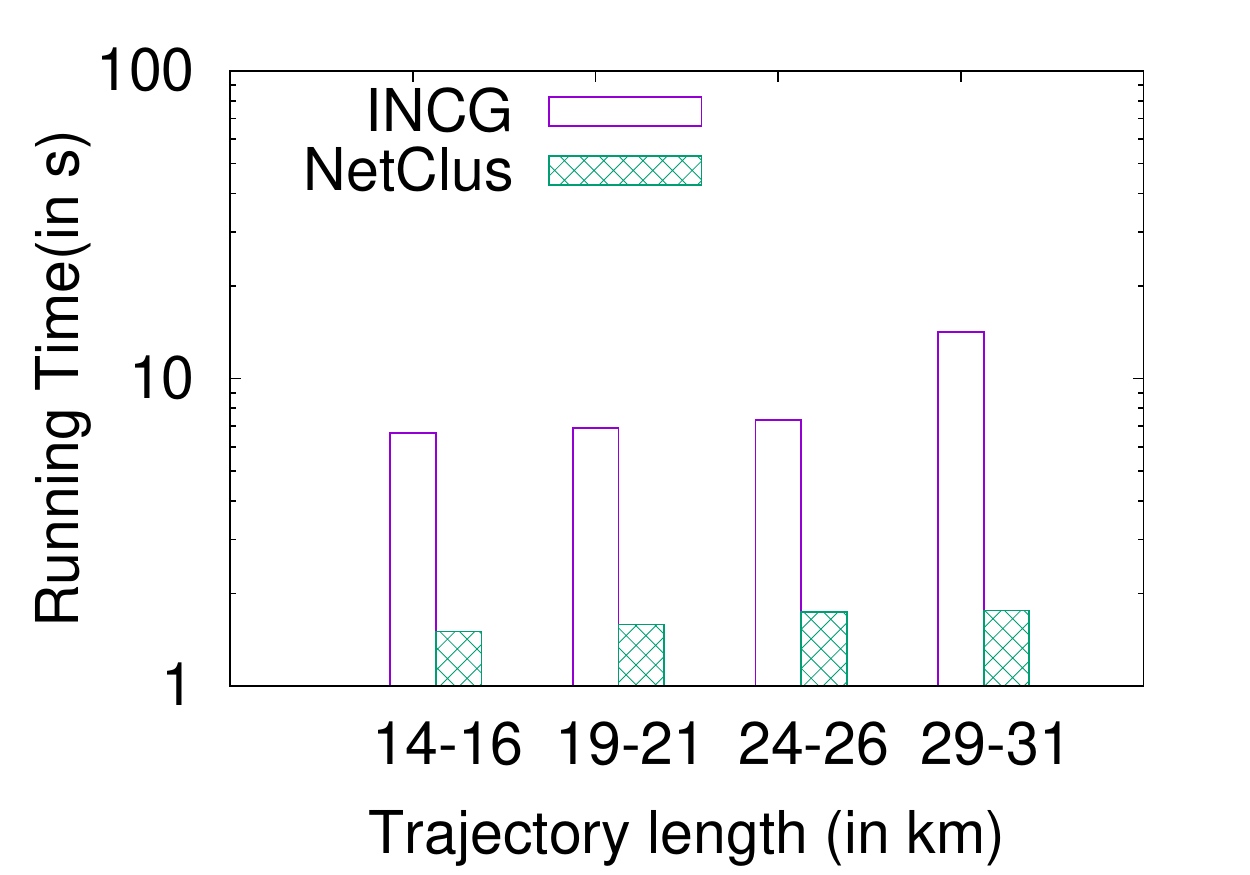}
		\label{subfig:lengthtime}
	}
	\figcaption{Effect of trajectory length ($k=5$ and $\tau=0.8$ Km).}
	\label{fig:length}
\end{figure}

\textbf{Number of Sites and Trajectories:}
Fig.~\ref{fig:scale} shows scalability results with varying number of candidate sites and trajectories on the \bl dataset. \nc is about an order of magnitude faster than \inc.


\textbf{City Geometries:}
We experimented with three typical city geometries, Atlanta, New York, and
Bangalore (Fig.~\ref{fig:synthetic}).  New York has a star
topology while Bangalore is poly-centric.  Consequently, Bangalore has a larger
utility percentage.  Since Atlanta has a mesh structure with trajectories
distributed all over the city, its utility is lowest.  There is not much
difference in the running times, though. Bangalore has least running time due to its smaller road network.


\textbf{Length of Trajectories:}
To determine the effect of length of trajectories, the trajectories were divided
into four classes based on their lengths and from each of the classes, 5,000
trajectories were sampled (Fig.~\ref{fig:length}).  Longer
trajectories are easier to cover since they pass through more number of
candidate sites over a larger area and, therefore, exhibit higher utility than
the shorter ones.  The running time also increases with trajectory length due to
more number of update operations of the marginal utilities. 

\subsection{Index Construction}


\begin{table}[hbt]
\scriptsize
	\begin{center}
	\begin{tabular}{rrrrrr}
	$R_p$ (Km) & $\kappa$ & $\bar{|\Lambda|}$ & $\bar{|\tl|}$ & $\bar{|\cl|}$ & Run-time (s)	\\
	\hline
0.0093	&	258340	&	1.04	&	561.88	&	4.29	&	269.35	\\
0.0286	&	195910	&	1.15	&	571.41	&	6.43	&	239.58	\\
0.0163	&	233729	&	1.38	&	592.45	&	12.62	&	255.60	\\
0.0500	&	153210	&	1.76	&	626.18	&	19.44	&	221.79	\\
0.0875	&	112223	&	2.40	&	675.60	&	26.04	&	204.38	\\
0.1531	&	76836	&	3.51	&	757.03	&	45.38	&	192.12	\\
0.2680	&	48288	&	5.58	&	895.19	&	64.27	&	188.25	\\
0.4689	&	28510	&	9.46	&	1162.93	&	53.74	&	205.74	\\
0.8207	&	15775	&	17.10	&	1525.09	&	42.73	&	281.62	\\
1.4361	&	8258	&	32.66	&	2162.01	&	35.62	&	537.27	\\
2.5133	&	4202	&	64.18	&	3092.64	&	21.73	&	1300.12	\\
4.3982	&	2024	&	133.24	&	4148.03	&	12.49	&	3231.92	\\
7.6968	&	938	&	287.51	&	7537.33	&	3.51	&	7333.60	\\
	\hline 
	\end{tabular}
	\tabcaption{Details of indexing for Beijing road network comprising of 269,686 nodes with $\eps=0.75$.}
	\label{tab:cluster}
\end{center}
\end{table}

Referring to Table~\ref{tab:cluster}, we observe that as the cluster radius
$R_p$ increases, the number of clusters $\eta_p$ decreases as the average
dominating set sizes $\bar{|\Lambda|}$ increases.  Therefore, the average
number of trajectories passing through a cluster $\bar{|\tl|}$ also increases.
The average number of neighbors of a cluster $\bar{|\cl|}$ initially increases
but finally decreases.  
Importantly, we observe that the
\emph{offline} index construction times across different cluster radii are
practical. 


\section{Related Work}
\label{sec:related}


The related work falls in two main classes, \emph{optimal location} queries
\cite{du2005optimal, ghaemi2010optimal, xiao2011optimal,
chen2014efficient,li2013trajectory}, and \emph{flow based facility location}
problems \cite{berman1992optimal, berman1995locatingDiscretionary,
berman1995locatingMain,berman1995locating,
berman2002generalized,berman1998flow}.  

An optimal location (OL) query has three inputs: a set of facilities, a set  of
users, and a set of candidate sites.  The objective is to determine a candidate
site to set up a new facility that optimizes an objective function based on the
distances between the facilities and the users.  Comparing OL query with \tops
query, we note that: (a)~While \emph{fixed} user locations are considered for OL
queries, \tops uses \emph{trajectories} of users. (b)~OL queries report only a
single optimal location, while \tops reports $k$ locations.  (c)~Unlike OL
queries that are solvable in polynomial time, \tops is NP-hard (it is
polynomially solvable only for $k=1$).

Recently, \cite{li2013trajectory} studied OL queries on trajectories over a road
network.  Two algorithms were proposed to compute the optimal road segment to
host a new service. Their work has quite a few limitations and differences when
compared with our work: (a)~Since a single optimal road segment is reported,
their problem is polynomially solvable.  (b)~Their work identifies the optimal
road segment, rather than the optimal site.  (c)~There is no analysis on the
quality of the reported optimal road segment, either theoretically or
empirically.  (d)~It is not shown how does the reported road segment performs
for other established metrics, such as number of new users covered, distance
traveled by the users to avail the service, etc.


Facility location problems \cite{drezner1995facility, FacilityLocation}
typically consider a set  of users, and a set of candidate sites.  The goal is
to identify a set of $k$ candidate sites that optimize certain metrics such as
covering maximum number of users, or minimizing the average distance between a
user and its nearest facility, etc. Almost all of these problems are NP-hard.
While early works assumed that the users are static, mobile users are now
considered.  The \emph{flow based} facility location works
\cite{berman1992optimal, berman1995locatingDiscretionary,
berman1995locatingMain,berman1995locating, berman2002generalized,berman1998flow}
assume a flow model to characterize human mobility, instead of using real
trajectories.
A fairly comprehensive literature survey is available in \cite{FIFLPSurvey}.  We
briefly outline the major works related to the different  versions of \tops
queries that have been discussed in Sec.~\ref{sec:variants}.
In \cite{berman1992optimal}, few exact and approximate algorithms for \topsone
and \topsfour, were presented, under the restriction that a customer would stay
on the path, i.e., $\tau=0$.  Later, in \cite{berman1995locatingMain}, few
generalizations of the model were proposed, where the customers were allowed to
deviate.  These include \topsone, \topstwo and \topsthree problems.  Further
generalizations of \topsone were examined in \cite{berman1995locating}, with
probabilistic customer flows, single and multiple flow interceptions, and fixed
and variable installation costs of the services. Several existing flow-based facility location models were generalized in \cite{zeng2010generalized}.
 Few  flow refueling location models have been proposed in \cite{lim2010heuristic,kuby2009optimization,mirhassani2012flexible}
for sitingalternative- fuel stations. We have already discussed how our work differs from these works in
Sec.~\ref{sec:Intro}. 

Recently, a greedy heuristic was proposed for the \topsone problem in  \cite{li2016mining}. under the assumption that $\tau=0$, i.e., no customer deviations are allowed. We present a detailed study of the \topsone problem in \cite{tops_icde}, with arbitrary value of $\tau$.

\vspace{-0.10in}
\section{Conclusions}
\label{sec:concl}

In this paper, we have proposed a generic \tops framework to solve the problem
of finding facility locations for trajectory-aware services. We showed that the problem is NP-hard, and proposed a greedy heuristic with constant factor approximation bound. However, this heuristic does not scale for large datasets. Thus, we  developed
an index structure, \nc, to make it practical for city-scale road networks.
Extensive experiments over real datasets showed that \nc yields solutions that
are comparable with those of the greedy heuristic, while being significantly faster, and low in memory overhead. The proposed framework can handle a wide class of objectives, and additional constraints, thus making it highly generic and practical.


{
\bibliographystyle{abbrv}
\balance
\bibliography{papers}
}

\pagebreak

\appendix

\section{Algorithms for \tops}

\subsection{Linearization for Optimal Algorithm}
\label{app:ip}

Each inequality in Ineq.~\eqref{eq:MaxUtilityConstraint} can be expressed as
$$U_j \leq \max\{U_{j1}, U_{j2}\}$$
where
$$U_{j1} \leq \max\{\pref(T_j,s_i)x_i|1 \leq i \leq \lfloor n/2 \rfloor\}$$
and
$$U_{j2} \leq \max\{\pref(T_j,s_i).x_i|\lfloor n/2 \rfloor+1 \leq i \leq n\}$$

The terms $U_{j1}$ and $U_{j2}$ can be recursively expressed in the same manner
as in case of $U_j$.

Finally, the constraint $U_j \leq \max\{U_{j1},U_{j2}\}$ can be linearized as
follows:
\begin{align*}
U_{j1} &\leq U_{j2} + M.y_j \\
U_{j2} &\leq U_{j1} + M.(1-y_j) \\
U_j &\leq U_{j2} + M.y_j \\
U_{j} &\leq U_{j1} + M.(1-y_j) \\
y_j &\in \{0,1\}
\end{align*}
where $M$ is a sufficiently large number.

\subsection{Inc-Greedy Algorithm}
\label{app:inc-greedy}

\begin{algorithm}[t]
	\caption{\incg}
	\label{algo:incg}
\begin{algorithmic}[1]
{
\scriptsize
\Procedure{\incg}{$k,\tau,\pref$}
\State Compute the sets $\tc,\stc$ and the site-weights.
\State $\mathcal{Q}_0 \leftarrow \varnothing$
\ForAll{$ s_i \in \mathcal{S}$} 
\State $U_0(s_i) \leftarrow w_i$
\State $\forall T_j \in \tc(s_i), \alpha_{ji} \leftarrow \pref(T_j,s_i)$
\EndFor
\State $\forall T_j \in \mathcal{T}, U_j \leftarrow 0$
\ForAll{$\theta=1,\dots,k$} \label{line:for loop of iterations}
\State $s_\theta\leftarrow\arg\max_{\forall s_i\in\mathcal{S} \setminus \mathcal{Q}_{\theta-1} }\{ U(\mathcal{Q}_{\theta-1}\cup\{s_i\} )-U(\mathcal{Q}_{\theta-1} )\}$
\State $\mathcal{Q}_\theta \leftarrow \mathcal{Q}_{\theta-1} \cup \{s_\theta\}$
\ForAll{trajectory $T_j \in \tc(s_\theta)$} \label{line:for loop of trajectories}
\State $U_j \leftarrow \max(U_j,\pref(T_j,s_\theta))$ \label{line:new Utility of trajectory j}
\If{$U_j$ changes}
\ForAll{site $s_i \in \stc(T_j), s_i \notin \mathcal{Q}_{\theta}$} \label{line:for loop of sites}
\State $\alpha_{ji}^\prime \leftarrow \max(0,\pref(T_j,s_i)-U_j)$ \label{line:delta utility}
\State $U_\theta(s_i)\leftarrow U_{\theta-1}(s_i)-(\alpha_{ji}-\alpha_{ji}^\prime)$
\State $\alpha_{ji}\leftarrow \alpha_{ji}^\prime$
\EndFor
\EndIf
\EndFor
\EndFor
\EndProcedure
}
\end{algorithmic} 
\end{algorithm}

The \incg procedure is outlined in Algorithm~\ref{algo:incg}.  The algorithm
begins by computing the sets $\tc,\stc$, and the site-weights. It then
initializes the marginal utilities of each site $s_i \in \mathcal{S}$ to its
site weight, i.e., $U_0(s_i) = w_i$.  Further, for each trajectory $T_j \in
\tc(s_i)$, it initializes $\alpha_{ji} = \pref(T_j,s_i)$.

In iteration $\theta=1,\dots,k$, it selects the site $s_{\theta}$ with
\emph{maximal} marginal utility, i.e., $\forall s_i \notin
\mathcal{Q}_{\theta-1}, U_\theta(s_\theta) \geq U_\theta(s_i)$.  If multiple
candidate sites yield the same maximal marginal gain we choose the one with
maximal weight. Still, if there are ties, then without loss of generality, the
site with the highest index is selected.

For each trajectory $T_j \in \tc(s_\theta)$ (line~\ref{line:for loop of
trajectories}), first the utility $U_j$ is updated (line~\ref{line:new Utility
of trajectory j}).  In case there is a gain in $U_j$, each site $s_i \in
\stc(T_j)$ is processed as indicated in line~\ref{line:for loop of sites}.  In
line~\ref{line:delta utility}, the new marginal utility of trajectory $T_j$
w.r.t. site $s_i$ is computed.  Using this value, the marginal utility of site
$s_i$ and $\alpha_{ji}$ are updated.

\section{NetClus}
\label{app:nc}

\subsection{Clustering based on Jaccard Similarity}
\label{app:Jaccard}

\begin{table}[t]
\scriptsize
	\begin{center}
	\begin{tabular}{c|c|c}
	\hline
	$\tau$ (Km) & Running Time (s) & Memory (in GB)\\
	\hline
0	&	2890	&	3.88	\\
0.2	&	12473	&	8.93	\\
0.4	&	26693	&	9.98	\\
0.8	&	28391	&	14.23	\\
1.2	&	29058	&	16.57	\\
1.6	&	32647	&	20.85	\\
2.0	&	34618	&	26.74	\\
2.4 & \multicolumn{2}{c}{Out of memory}\\
\hline 
	\end{tabular}
	\tabcaption{Performance of Jaccard similarity based clustering for Beijing road network using a Jaccard distance threshold $\jd=0.8$.}
	\label{tab:jaccard}
\end{center}
\end{table}

Given two nodes $s_i,s_j \in \s$, the Jaccard similarity between their trajectory covers is
\begin{align*}
J_s(s_i,s_j) &= \frac{|\tc(s_i) \cap \tc(s_j)|}{|\tc(s_i) \cup \tc(s_j)|}
\end{align*}
The Jaccard distance is
\begin{align*}
	J_d(s_i,s_j)=1-J_s(s_i,s_j)
\end{align*}

The nodes in the road network $V$ are clustered using the above measure as follows.

The node $v_1$ with the highest weight, i.e., the one whose sum of preference
scores over the covered trajectories as defined in Sec.~\ref{sec:approx}, is
chosen as the first cluster center of the cluster $g_1$. All nodes that are
within a threshold of $\jd$ Jaccard distance from $v_1$ are added to $g_1$.
From the remaining set of unclustered nodes, the node with the highest weight
is chosen as a cluster center, and the process is repeated until each node is
clustered.



The results of clustering on the \bl dataset are shown in
Table~\ref{tab:jaccard}.
The limitations of Jaccard similarity based clustering are explained in Sec.~\ref{sec:offline}.

\subsection{Complexity Analysis of Greedy-GDSP}
\label{app:gdsp}

\begin{thm}
	Greedy-GDSP runs in $O(|V| . (\nu \log \nu + \eta))$ time where $\nu$ is the
	maximum number of vertices that are reachable within the largest round-trip
	distance $R_{max}$ from any vertex $v$, and $\eta$ is the number of clusters
	returned by the algorithm.
\end{thm}

\begin{proof} 
As the underlying graph
models a road network which is roughly planar, we assume that the number of
edges (i.e., road segments) incident on any set of $\nu$ vertices is $O(\nu)$.
Since the cost of running the shortest path algorithm from a given node in a graph $G = (V,E)$ is
$O(|E| \log |V|)$ \cite{algorithmsbook}, therefore, the initial distance
computation for a given source vertex requires $O(\nu \log \nu)$ time. Thus, the distance computation across all the vertices in $V$ takes $O(|V|.\nu \log \nu)$. 
The neighbors of each node $v$ can be maintained as a list sorted by the round-trip
distance from $v$.  Therefore, computing the dominating set for a particular $R$
requires at most $\nu$ time for each vertex.  The total time for the construction of dominating sets, hence, is $O(|V|.\nu)$.

Choosing the vertex with the largest \emph{incremental} dominating set requires
bitwise OR operations of FM sketches.  For $|V|$ vertices, there are at most $|V|-1$
such operations, each requiring $O(f)$ time (since there are $f$ copies of FM
sketches each with a constant size of $32$ bits).  As the number of clusters
produced is $\eta$, the running time is $O(|V|. \eta)$.

The total running time is, therefore, $O(|V| . (\nu \log \nu + \eta))$.
\hfill{}
\end{proof}

\subsection{Complexity of NetClus}
\label{app:netclus-complexity}

\begin{thm}
	\label{appthm:complexity}
	The time and space complexities of \nc are $O(k.\eta_p.\xi_p)$ and
	$O(\sum_{p=1}^t (\eta_p(\xi_p + \lambda_p)))$ respectively.
\end{thm}

\begin{proof}
	The number of cluster representatives, $|\widehat{\mathcal{S}}|$ is at most
	the number of clusters, $\eta_p$. For any cluster $g_i \in \inst_p$,
	$|\tl(g_i)| = O(\xi_p)$.  We assume that the number of neighboring
	clusters for any cluster is bounded by a constant, i.e., $|\cl(g_i)| =
	O(1)$.
	(Table~\ref{tab:cluster} that lists typical mean values of
	$|\cl(g_i)|$ empirically supports the assumption.) 
	Hence, computing the set $\widehat{\tc}(r_i)$
	requires at most $O(\xi_p)$ time.  The total time across all the $\eta_p$
	clusters, therefore, is $O(\eta_p. \xi_p)$.  The inverse maps
	$\widehat{\cc}(T_j)$ for the trajectories can be computed along with
	$\widehat{\tc}(r_i)$.  Hence, the total time remains $O(\eta_p. \xi_p)$.
	Using Th.~\ref{thm:complexity of incg}, the subsequent \incg phase requires
	$O(k.\eta_p.\xi_p)$ time.  Therefore, the overall time complexity of the
	algorithm is $O(k.\eta_p.\xi_p)$.
  
	We next analyze the space complexity for a particular index instance
	$\inst_p$.  Storing the sets $\tl(g_i)$, each of size at most $\xi_p$,
	across all the $\eta_p$ clusters require $O(\eta_p.\xi_p)$ space. Storing
	the inverse maps $\cc(T_j)$ for all the trajectories, thus, also requires
	$O(\eta_p.\xi_p)$ space.  As there are at most $\lambda_p$ vertices in any
	cluster, storing their ids and distances from the cluster center requires
	$O(\lambda_p)$ space for each cluster.  Assuming $|\cl(g_i)| = O(1)$, the
	total storage cost for a cluster is $O(\lambda_p)$.  Therefore, the total
	storage cost is $O(\eta_p.(\xi_p + \lambda_p))$.  Summing across all the
	index instances, the entire space complexity is $O(\sum_{p=1}^t
	(\eta_p(\xi_p + \lambda_p)))$.
	\hfill{}
\end{proof}

\end{document}